\documentclass[a4paper,UKenglish,cleveref, autoref, thm-restate]{lipics-v2021}
\usepackage{graphicx} % Required for inserting images
\usepackage{amsmath,amsthm,amssymb}
\usepackage{stmaryrd}
\usepackage{mathtools}
\usepackage{algorithm}
\usepackage{tikz}
\usepackage{todonotes}
\usepackage{lineno}
\usepackage{algpseudocode}
\usepackage[bold,basic]{complexity}
\usetikzlibrary{automata, positioning, arrows,calc}
\usepackage[bibliography=common]{apxproof}

\pdfoutput=1 %uncomment to ensure pdflatex processing (mandatatory e.g. to submit to arXiv)
\hideLIPIcs  %uncomment to remove references to LIPIcs series (logo, DOI, ...), e.g. when preparing a pre-final version to be uploaded to arXiv or another public repository

\title{A Theory of Hanoi Omega-Automata and Games}

\author{Emmanuel Filiot}{Universit\'e libre de Bruxelles, Belgium}{emmanuel.filiot@ulb.be}{0000-0002-2520-5630}{}
\author{Allen Joseph}{Universit\'e libre de Bruxelles, Belgium}{}{}{}
\author{Guillermo A. P\'erez}{University of Antwerp -- Flanders Make,Belgium}{guillermo.perez@uantwerpen.be}{0000-0002-1200-4952}{}
\author{Saina Sunny}{Universit\'e libre de Bruxelles, Belgium}{saina.sunny@ulb.be}{0009-0005-1366-0168}{}

\authorrunning{E. Filiot, A. Joseph, G. A. P\'erez, and S. Sunny} 
 


\category{} %optional, e.g. invited paper

\relatedversion{} %optional, e.g. full version hosted on arXiv, HAL, or other respository/website
\funding{This work was supported by the Belgian FWO-FNRS ``SynthEx'' (G0AH524N) project.}%optional, to capture a funding statement, which applies to all authors. Please enter author specific funding statements as fifth argument of the \author macro.

\acknowledgements{We would like to thank Moshe Vardi for asking the main questions answered in this work and Alexandre Duret-Lutz for discussing HOA automata with us.}%optional

\nolinenumbers %uncomment to disable line numbering

\newcommand{\val}{\mathrm{Val}}
\newcommand{\ap}{\mathcal{P}}
\newcommand{\lan}[1]{\mathcal{L}(\mathcal{#1})}
\newcommand{\sem}[1]{\llbracket #1 \rrbracket}
\newcommand{\Inf}{\mathit{Occ}_\infty}
\newcommand{\Fin}{\mathit{Occ}_{<\infty}}
\newcommand{\Acc}{\mathit{Acc}}
\newcommand{\lassowordlimit}{\size{\mathcal A'} 2^{\size{\mathcal B'}^2\log\size{\mathcal B'}}}
\newcommand{\size}[1]{\lVert #1 \rVert}
\newcommand{\inp}{\textsf{in}}
\newcommand{\outp}{\textsf{out}}

\newrobustcmd{\saina}[2][]{{\color{black}\todo[color=blue!20,noshadow,#1]{{\bf Saina:} #2}}\ignorespaces}

\appendixprelim{}

\begin{document}

\maketitle
\begin{abstract}
    The Hanoi Omega-Automata (HOA) format has established itself as the definitive standard for encoding $\omega$-regular automata in modern synthesis tools. While HOA is widely adopted due to its succinct symbolic representation, using Boolean formulas as transition guards and transition-based coloring, the exact computational cost of these features has remained understudied. This paper provides the first systematic investigation into the theoretical complexity of decision problems for HOA-encoded automata and games. We establish that the structural features of HOA, specifically the symbolic encoding of large alphabets, make classical problems more complex than in traditional formats. We prove that the non-emptiness problem is {NP-complete} for all standard acceptance conditions, with hardness arising directly from the Boolean transition guards. For language inclusion, we show the problem is {PSPACE-complete} for most conditions but reaches {EXPSPACE-completeness} for Emerson-Lei acceptance. Furthermore, we formalize {Hanoi Omega-Games (HOG)}, where the underlying arena is a deterministic HOA with atomic propositions partitioned into inputs and outputs. We provide tight complexity bounds for solving HOGs, ranging from {$\Pi_2$-completeness} for parity and safety conditions to {PSPACE-completeness} for Muller and Emerson-Lei objectives. Finally, we generalize our techniques to {symbolic games} where transitions are guarded by formulas in arbitrary decidable first-order theories.
\end{abstract}

\section{Introduction}

The Hanoi Omega-Automata (HOA) format has established itself as the definitive standard for encoding $\omega$-regular automata. Its widespread adoption is evident in its integration into cornerstone tools like Spot, Owl, and LTL3TELA \cite{spot,owl,ltl3tela}. We note that the foundational paper \cite{hoa} introducing the format has already earned more than 150 citations.

Automata encoded using HOA are a critical component of the temporal synthesis pipeline. Synthesizing controllers from Linear Temporal Logic (LTL) specifications typically requires compiling formulas into automata and then solving a two-player zero-sum game. In this process, the HOA-encoded automaton is what ultimately determines the winner.
Modern research continues to refine the automata-theoretic pipeline to achieve real-world scalability. Although the worst-case complexity of temporal synthesis is known to be \textsc{2ExpTime}-complete \cite{synthesis_handbookmc}, the annual SYNTCOMP competition \cite{DBLP:journals/sttt/JacobsPABCCDDDFFKKLMMPR24} highlights an active, ongoing effort to optimize these tools and benchmarks for practical applications.

A significant gap exists in our understanding of the theoretical complexity of analyzing HOA-encoded automata. Despite being a standard input format for SYNTCOMP, the exact complexity of automata- and game-theoretic questions specifically for HOA-encoded automata has not been formally studied. Personal communications with Alexandre Duret-Lutz (author of the Spot library \cite{spot}) confirm that while these problems are routinely solved---often using binary decision diagrams \cite{bdds_handbookmc}---their precise complexity bounds remain open.

\subparagraph{What is different about HOA-encoded automata?}
HOA introduces unique structural features that differ from classical set-based definitions of automata. Notably, it utilizes transition-based acceptance rather than state-based acceptance (for example, priority-labeled transitions, instead of states, for parity acceptance conditions). Furthermore, it employs a symbolic alphabet encoded succinctly by labeling transitions with Boolean formulas over atomic propositions. This encoding allows a single transition to represent all (out of a possibly exponential set, in the number of atomic propositions) transitions on a valuation that satisfies the formula. We note that this makes the model similar to \emph{symbolic automata} over the Boolean domain~\cite{symbauto}. In fact, HOA-encoded automata are a concrete instantiation of symbolic automata over infinite words with propositional logic as a fixed Boolean algebra.

These subtle structural differences have a profound impact on the computational complexity of classical automata-theoretic problems. Recent work suggests that moving from state-based to transition-based acceptance alone can change the complexity of certain questions (see \cite{casares-beatcs} and references therein), and it is intuitive that the succinct encoding of large alphabets via formulas would have similar effects.

\subparagraph{What about Hanoi Omega-Games?}
Since the 2020 edition of SYNTCOMP, a synthesis track starting from HOA---extended with a partitioning of the the atomic propositions into inputs and outputs---was added \cite{DBLP:journals/sttt/JacobsPABCCDDDFFKKLMMPR24}. At the time, most synthesis solvers for the temporal-logic specification track were translating input specifications into HOA-encoded automata and then solving a game played on it. The new track on HOA games (henceforth shortened as HOGs) was intended to encourage the development of tools that focus on the latter part of the synthesis pipeline.
One of the only alternative formats for games at the time was the one used by the PGSolver collection of parity-game generators and solvers \cite{friedmann2009pgsolver}. The PGSolver format encodes a (parity) game by explicitly enumerating edges of the bipartite graph on which the game is played. Naively translating a HOG into such an explicit representation incurs an exponential blow-up.

While HOG-solving is practically possible using Knor ~\cite{DBLP:conf/tacas/DijkAT24} or ltlsynt~\cite{spot}, there has been no systematic study of the complexity of solving HOGs. Moreover, while we can place HOA within the framework of symbolic automata, we are not aware of a general framework where HOGs fit, i.e. some sort of \emph{symbolic game} framework. The closest match may be \cite{DBLP:conf/cav/HeimD25,DBLP:conf/cav/AzzopardiSPS25} (and models referenced therein). However, the referenced models all have variable valuations that affect the state of the system while HOGs have a much more local parametric structure: valuations are chosen to select one of finitely many successors.

Figure~\ref{fig:HOG} depicts a deterministic HOA over atomic propositions $r$ and $g$ (request and grant), encoding the property $G(r \rightarrow F(g)) \land G(g \rightarrow X(\lnot g))$. This means every request is eventually granted, and grants cannot occur in consecutive steps. Transitions are labeled with Boolean formulas as well as colors (here $1$ and $2$), and the acceptance condition is defined to be $\mathsf{Inf}(\{1\})$, meaning that color $1$ must be seen infinitely often. This HOA can be viewed as an HOG where the environment controls the proposition $r$ and controller controls $g$, and the controller aims to satisfy the property, i.e., visit color $1$ infinitely often. Controller wins this game by applying any strategy which in state $q_2$, whatever be the input, set $g$ to \textsf{False} and state $q_3$, whatever be the input, set $g$ to \textsf{True}. 
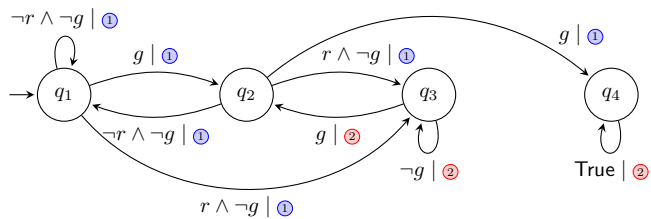
\begin{figure}
    \centering
     \begin{tikzpicture}[shorten >=1pt, auto, node distance=3cm, on grid, >=stealth, scale=0.8, every node/.style={scale=0.8}]

% States
\node[state, initial,initial text=] (q0) {$q_1$};
\node[state, right of=q0] (q1) {$q_2$};
\node[state, right of=q1] (q2) {$q_3$};
\node[state, right of=q2] (q3) {$q_4$};

% Transitions
\draw[->] (q0) edge[loop above] node {$\lnot r \land \lnot g \mid \tikz[baseline=(n.base)]{
\node[circle, draw=blue, fill=blue!20, inner sep=1pt, font=\scriptsize] (n) {1};
}$} (q0);
\draw[->] (q2) edge[loop below] node {$\lnot g \mid \tikz[baseline=(n.base)]{
\node[circle, draw=red, fill=red!20, inner sep=1pt, font=\scriptsize] (n) {2};
}$} (q2);
\draw[->] (q0) edge[bend left=20] node[above] {$g \mid \tikz[baseline=(n.base)]{
\node[circle, draw=blue, fill=blue!20, inner sep=1pt, font=\scriptsize] (n) {1};
}$} (q1);
\draw[->] (q1) edge[bend left=20] node[below] {$\lnot r \wedge \lnot g \mid \tikz[baseline=(n.base)]{
\node[circle, draw=blue, fill=blue!20, inner sep=1pt, font=\scriptsize] (n) {1};
}$} (q0);
\draw[->] (q1) edge[bend left=20] node[above,xshift=0.5cm] {$r \land \lnot g \mid \tikz[baseline=(n.base)]{
\node[circle, draw=blue, fill=blue!20, inner sep=1pt, font=\scriptsize] (n) {1};
}$} (q2);
\draw[->] (q2) edge[bend left=20] node[below] {$g \mid \tikz[baseline=(n.base)]{
\node[circle, draw=red, fill=red!20, inner sep=1pt, font=\scriptsize] (n) {2};
}$} (q1);
\draw[->] (q1) edge[bend left=40] node[below, xshift=2.5cm] {$g \mid \tikz[baseline=(n.base)]{
\node[circle, draw=blue, fill=blue!20, inner sep=1pt, font=\scriptsize] (n) {1};
}$} (q3);

\draw[->] (q3) edge[loop below] node {$\textsf{True} \mid \tikz[baseline=(n.base)]{
\node[circle, draw=red, fill=red!20, inner sep=1pt, font=\scriptsize] (n) {2};
}$} (q3);
\draw[->] (q0) edge[bend right=50] node[below] {$r \land \lnot g \mid \tikz[baseline=(n.base)]{
\node[circle, draw=blue, fill=blue!20, inner sep=1pt, font=\scriptsize] (n) {1};
}$} (q2);

\end{tikzpicture}
    \caption{Hanoi omega-automaton underlying an Hanoi omega-game where the environment (resp. controller) controls $r$ (resp. $g$).}
    \label{fig:HOG}
    \vspace{-5mm}
\end{figure}

\subparagraph{Contributions.} In this work, we confirm that HOA-encoded $\omega$-automata present harder decision problems. We establish that classical problems such as emptiness and language inclusion are more complex in this format. Our most important contributions, though, concern HOGs: We provide tight complexity bounds for solving important classes of HOGs, directly addressing the setup used in the SYNTCOMP competition. Finally, we generalize our techniques to \emph{symbolic games} (over possibly infinite data domains) as hinted at before, e.g., to games where the controller and the environment choose values of natural vectors. 
The following table summarizes our results.

\begin{table}[!htb]
\centering
%\caption{Complexity of decision problems for Hanoi omega-automata and games.}
\renewcommand{\arraystretch}{1.25}
\begin{tabular}{|l|c|c|c|}
\hline
\textbf{Acceptance Condition} 
& \textbf{Emptiness} 
& \textbf{Language Inclusion} 
& \textbf{Game Solving} \\
\hline

\multicolumn{1}{|l|}{\shortstack[l]{\\ Reachability / Safety / \\ Büchi / co-Büchi / \\ Parity / Streett}}
& \multicolumn{1}{c|}{\NP-complete}
& \multicolumn{1}{c|}{\PSPACE-complete}
& \multicolumn{1}{c|}{$\boldsymbol{\Pi_2}$-complete} \\

\hline

Rabin
& \NP-complete
& \PSPACE-complete
& $\boldsymbol{\Sigma_3}$ and $\boldsymbol{\Pi_2}$-hard \\

\hline

Muller
& \NP-complete
& \PSPACE-complete
& \PSPACE-complete \\

\hline

Emerson--Lei
& \NP-complete
& \EXPSPACE-complete
& \PSPACE-complete \\

\hline
\end{tabular}

\end{table}

\section{Preliminaries}
We write $\mathbb{N} = \{1,2,\ldots\}$ for the set of natural numbers and $\mathbb{B} =\{\mathsf{True},\mathsf{False}\}$ for Boolean truth values. For every $d \in \mathbb{N}$, we use $[d]$ to denote the set $\{1,2,\ldots, d\}$.

\subparagraph*{Words, Alphabets.} A word (resp. $\omega$-word) over a (finite) alphabet is a finite (resp. infinite) sequence of symbols from the alphabet. The set of finite (resp. $\omega$-word) over an alphabet $\Sigma$ is denoted by $\Sigma^*$ (resp. $\Sigma^\omega$). The length of a word $w$ is denoted by $|w|$ (it is infinite for $\omega$-words).
The $i$th letter of $w$ is denoted by $w[i]$. We denote by $\mathit{elem}(w)$ (resp. $\Inf(w)$) the set of symbols in $\Sigma$ that occur in $w$ (resp. occur infinitely many times in $w$). In particular, $\Inf(w)=\emptyset$ if $w$ is a finite word. 
Similarly, we define $\Fin(w)$ to be the set of symbols that occur only finitely many times in $w$, i.e., $\Fin(w) = \mathit{elem}(w) \setminus \Inf(w)$.

\subparagraph*{Boolean Formula and Valuations.} The set of Boolean formulas over a set of atomic propositions $\ap$, denoted by  $\mathbb{B}(\ap)$, is defined in a standard way. The set of valuations $v : \ap \to \mathbb{B}$ of propositions is denoted by $\val(\ap)$, and given a formula $\varphi\in \mathbb{B}(\ap)$, we write $v\models \varphi$ if $\varphi$ evaluates to $\mathsf{True}$ under $v$, and denote by $\size{\varphi}$ the size of $\varphi$ (its number of symbols).

\subparagraph{Polynomial Hierarchy.}
We recall the definition of the polynomial hierarchy. Let $\boldsymbol{\Sigma_0} = \P$ be the class of languages decidable in deterministic polynomial time. For every $i>0$, let $\boldsymbol{\Sigma_i} = \NP^{\boldsymbol{\Sigma_{i-1}}}$,
that is, the class of languages decidable by a non-deterministic polynomial-time Turing machine equipped with an oracle for some language in $\boldsymbol{\Sigma_{i-1}}$. In particular, $\boldsymbol{\Sigma_1} = \NP$. Similarly, let $\boldsymbol{\Pi_i} = \coNP^{\boldsymbol{\Sigma_{i-1}}}$. By standard closure properties of the hierarchy, for every $i \ge 0$ we have
$\boldsymbol{\Sigma_i} = \boldsymbol{\mathrm{co}\Pi_i}$ and $\boldsymbol{\Pi_i} = \boldsymbol{\mathrm{co}\Sigma_i}$. All these classes lie within \PSPACE. A complete problem for $\boldsymbol{\Sigma_k}$ is $\textsc{QSAT}^\exists_k$, which is the set of satisfiable quantified Boolean formulas of the form $\exists \overline{x_1}\forall \overline{x_2}\exists \overline{x_3}\;\dots\;\phi$, where $\phi$ is quantifier-free, the $\overline{x_i}$ are tuples of variables, and the quantifier prefix contains exactly $k{-}1$ quantifier alternations. In particular, $\textsc{QSAT}^\exists_1$ is the classical \textsc{SAT} problem. Define $\textsc{QSAT}^\forall_k$ analogously, except that the formula begins with a block of universal quantifiers. Then $\textsc{QSAT}^\forall_k$ is $\boldsymbol{\Pi_k}$-complete.

 \subparagraph*{Classical Two-player Games.} A classical (turn-based) two-player game is a tuple $\mathcal{G} = (V, V_\inp, V_\outp, E, C, {\cal C}, Acc)$,
where $V$ is a set of vertices partitioned into $V_\inp$ and $V_\outp$, the positions of Player \textsf{In} and Player \textsf{Out}, respectively. The relation $E \subseteq V \times V$ is the edge relation. The set $C$ is a set of colors, and ${\cal C}: V \to C$ is a coloring function that assigns a color to each vertex. The winning condition $Acc \subseteq C^\omega$ is defined over infinite sequences of colors. A \emph{play} in $\mathcal{G}$ is an infinite sequence of vertices $\pi = v_0 v_1 v_2 \dots \in V^\omega$ such that $(v_i,v_{i+1}) \in E$ for all $i \geq 0$. If $v_i \in V_\inp$ (resp. $V_\outp$), then Player \textsf{In} (resp. Player \textsf{Out}) chooses the successor $v_{i+1}$. The play induces a sequence of colors ${\cal C}(\pi) = {\cal C}(v_0){\cal C}(v_1){\cal C}(v_2)\dots \in C^\omega$. Player \textsf{Out} wins the play if ${\cal C}(\pi) \in Acc$, otherwise Player \textsf{In} wins. The \emph{arena} of $\mathcal{G}$ is the 
tuple $(V, V_\inp, V_\outp, E, C, {\cal C})$.

A \emph{strategy} for Player~\textsf{Out} is a function $\sigma : V^* V_{\outp} \to V$
such that $(v,\sigma(hv)) \in E$ for all histories $hv \in V^*V_{\outp}$. Strategies for Player \textsf{In} are defined analogously. A strategy for Player~\textsf{Out} is \emph{winning} from a vertex if, no matter how Player~\textsf{In} plays, every resulting play satisfies the winning condition $Acc$. Similarly, a strategy for Player \textsf{In} is winning from a vertex if every resulting play violates the winning condition $Acc$.
The \emph{winning region} of Player \textsf{Out}, denoted by $W_\outp$, is the set of vertices from which Player \textsf{Out} has a strategy to ensure a win. The winning region of Player \textsf{In} is $W_{\inp} = V \setminus W_{\outp}$. 

\section{Hanoi Omega-Automata} 
\label{sec:HOA}
A \emph{Hanoi omega-automaton} (HOA) is a tuple
$\mathcal{A} = (Q, I, \mathcal{P}, \Delta, \mathcal{C})$ where
\begin{itemize}
    \item $ Q $ is a finite set of states and $ I \subseteq Q$ the set of initial states,
    \item $\ap$ is a finite set of atomic propositions,
    \item $\Delta \subseteq Q \times \mathbb{B}(\ap) \times Q$ is the transition relation,
    \item $\mathcal{C} : \Delta \to [d]$ assigns to each transition a color $c \in [d]$, where $d \in \mathbb{N}$ is called the \emph{index} of ${\cal A}$.
\end{itemize}
 
In the whole paper, we assume wlog that $d$ is smaller than the number of transitions, i.e. $d\leq |\Delta|$. The \emph{size} of an HOA $\mathcal{A}$, denoted by $\size{\mathcal{A}}$, is defined as 
$\size{\mathcal{A}} = |Q|+|\ap|+\sum_{(q,\varphi,q') \in \Delta} \size{\varphi}$.  An HOA is \emph{deterministic} if $|I| = 1$ and for every state $q$, and valuation $v$, $|\{(q,\varphi,q') \in \Delta \mid v \models \varphi\}| \leq 1$ (at most one transition whose guard $\varphi$ is satisfied). 

An (infinite) run $\rho$ of $\mathcal{A}$ is a sequence %$q_1 \varphi_1 q_2 \varphi_2 q_3 \dots$ 
$q_1 \xrightarrow{\varphi_1} q_2 \xrightarrow{\varphi_2} \cdots$
of states and Boolean formulas such that $(q_i, \varphi_{i}, q_{i+1}) \in \Delta$ for all $i \geq 1$. An input word to $\mathcal{A}$ is an $\omega$-word $w = v_1 v_2 v_3 \dots$ over alphabet $\val(\ap)$. We write $w \models \rho$ if $v_i \models \varphi_i$ for all $i \geq 1$. The coloring function $\mathcal{C}$ is extended to runs as follows: $\mathcal{C}(\rho) = \mathcal{C}(q_1,\varphi_1,q_2)\mathcal{C}(q_2,\varphi_2,q_3)\dots$ denote the infinite sequence of colors that label transitions along the run. To be able to define the language of the automaton, we still need a way to determine which runs are \emph{accepting}. 
An \emph{acceptance condition}, denoted by $\Acc$ is a subset of $[d]^\omega$. A run $\rho = q_1 \xrightarrow{\varphi_1} q_2 \xrightarrow{\varphi_2} q_3 \cdots$ of $\mathcal{A}$ is \emph{accepting} (for $\Acc$) if $q_1 \in I$ and $\mathcal{C}(\rho) \in \Acc$.
The \emph{language} of $\mathcal{A}$, written $\mathcal{L}_{\Acc}(\mathcal{A})$ (just written $\lan{A}$ if $\Acc$ is clear from the context) is the set of input words $w$ for which the automaton has an accepting run $\rho$ such that $w \models \rho$. In this paper, we consider the acceptance conditions below.

\subparagraph*{Emerson-Lei (EL) acceptance.}
 An EL acceptance condition for an index $d$ is defined by an acceptance formula $\alpha$, which is a Boolean combination of atomic predicates $\mathsf{Inf(C)}$ and $\mathsf{Fin(C)}$, for any subset $C$ of colors. 
A color sequence $\bar{c}\in [d]^\omega$ satisfies $\mathsf{Inf(C)}$ if $\Inf(\bar{c}) \ \cap \ C \neq \emptyset$, and it satisfies $\mathsf{Fin(C)}$ if $\Inf(\bar{c}) \ \cap \ C = \emptyset$. These two interpretations are naturally extended to define the semantics of an acceptance formula $\alpha$, denoted $\sem{\alpha}\subseteq [d]^\omega$, as the set of color sequences satisfying $\alpha$.  

An Emerson-Lei condition is a language $Acc = \sem{\alpha}$ for some acceptance formula $\alpha$. The \emph{size} of the Emerson-Lei acceptance condition is the number of symbols occurring in its acceptance formula. 

The following are some classical $\omega$-automaton acceptance conditions. An acceptance condition $Acc$ for an index $d$ is 
\begin{itemize}
    \item B\"uchi if $Acc = \sem{\mathsf{Inf}(B)}$ for some $B \subseteq [d]$, and coB\"uchi if $Acc = \sem{\mathsf{Fin}(B)}$.
    \item Rabin if $Acc = \sem{\bigvee_{i \in [k]} (\mathsf{Inf}(E_i) \land \mathsf{Fin}(F_i))}$ for $k$ (Rabin) pairs $E_i,F_i\subseteq [d]$. 
    \item Streett if $Acc = \sem{\bigwedge_{i \in [k]} (\mathsf{Fin}(E_i) \lor \mathsf{Inf}(F_i))}$ for $k$ (Streett) pairs $E_i,F_i\subseteq [d]$. 
    \item Muller if $Acc = \sem{\bigvee_{i \in [k]} \bigwedge_{c \in E_i} \mathsf{Inf}(\{c\})}$ for $k$ sets $E_i\subseteq [d]$, called Muller sets.
    \item Parity (max even) if $Acc = \sem{\bigvee_{c \in [d], \text{$c$ is even}}  \mathsf{Inf}(\{c\}) \land \mathsf{Fin}(\{c+1, \ldots,d\})}$.
    \item Reachability if $Acc = \{ \bar{c} \mid \mathit{elem}(\bar{c}) \cap R \neq \emptyset \}$ for some $R \subseteq [d]$.
    \item Safety if $Acc = \{ \bar{c} \mid \mathit{elem}(\bar{c}) \subseteq S \}$ for some $S \subseteq [d]$.
\end{itemize}
Both reachability and safety acceptance can be expressed as Emerson-Lei conditions by a simple transformation of the automaton. The reachability acceptance can be reduced to B\"uchi by a simple transformation of the automaton: Let $q_{n}$ be a new sink state with a self-loop $q_n \xrightarrow{\mathsf{True} \mid c} q_n$ for all $c \in R$. For each transition $p \xrightarrow{\varphi \mid c}q$ where $c \in R$ in the original HOA, add a transition $p \xrightarrow{\varphi \mid c}q_{n}$. A run satisfies the reachability condition in the original HOA iff there exists a run in the transformed HOA that eventually stays within colors in $R$, which is equivalent to the B\"uchi condition $\mathsf{Inf}(R)$. Similarly, the safety acceptance can be reduced to co-B\"uchi by a simple transformation of the automaton --- Let $q_{n}$ be a new sink state with self loop $q_n \xrightarrow{\mathsf{True} \mid c} q_n$ for all $c \in [d] \setminus S$. Replace each transition $p \xrightarrow{\varphi \mid c}q$ where $c \not\in S$ in the original HOA, with a transition $p \xrightarrow{\varphi \mid c}q_{n}$. A run satisfies the safety condition in the original HOA iff the transformed HOA satisfies the Co-B\"uchi condition $\mathsf{Fin}([d] \setminus S)$.

The size of an acceptance condition $\Acc$ (amongst the latter) is defined as the size of the acceptance formula defining it, or $|R|$ (resp. $|S|$) for reachability (resp. safety). When $\Acc$ is clear from the context, then $\size{\cal A}$ refers to the total size including the size of $\Acc$.

\subparagraph*{State-Based Hanoi Omega-Automaton.} An HOA $\mathcal{A} = (Q, I, \ap, \Delta, \mathcal{C})$ is said to have a \emph{state-based coloring function} if all transitions that leave the same state are assigned the same color, i.e., for any state $q \in Q$, if $(q,\varphi,p), (q,\varphi',p') \in \Delta$ for some states $p,p' \in Q$ and Boolean formula $\varphi,\varphi' \in \mathbb{B}(\ap)$, then $\mathcal{C}(q,\varphi,p) = \mathcal{C}(q,\varphi',p')$. As a consequence, $\mathcal{C}$ depends only on the source state, and can therefore be defined as a function $\mathcal{C}: Q \to [d]$. For a run $\rho$ of $\mathcal{A}$, $\mathcal{C}(\rho)$ now denotes the sequence of colors that label the states along the run.
A \emph{state-based HOA} is an HOA equipped with such a state-based coloring function. 

\begin{lemma}\label{lem:to_state_based}
    For every HOA $\mathcal{A}$, one can construct a state-based HOA $\mathcal{B}$ such that, for any acceptance condition $\mathit{Acc}$, $L_{Acc}({\cal A}) = L_{Acc}({\cal B})$. Moreover, $\size{{\cal B}}$ is quadratic in $\size{\cal A}$.
\end{lemma}
\begin{proof}
    Let $\mathcal{A} = (Q_A, I_A, \mathbb{B}(\ap), \Delta_A, \mathcal{C}_A)$ with index $d$. 
    We can assume wlog that every state in $Q_A$ has an outgoing transition. For each state $q \in Q_A$, define:
      \[
        C(q) = \{\, \mathcal{C}_A(q, \varphi, q') \mid  (q, \varphi, q') \in \Delta_A\,\}.
    \]
    We construct $\mathcal{B} = (Q_B, I_B, \mathbb{B}(\ap), \Delta_B, \mathcal{C}_B)$ as follows:
    \begin{enumerate}
        \item The set of states is
        \[
            Q_B = \{\, (q,c) \mid q \in Q_A,\ c \in C(q) \,\}.
        \]
        \item The set of initial states is
        \[
            I_B = \{\, (q,c) \mid q \in I_A,\ c \in C(q) \,\}.
        \]
        \item The transition relation $\Delta_B$ is defined as the union of:
        \[
            \{\, ((q,c), \varphi, (q',c')) \mid (q,c), (q',c') \in Q_B, \ (q, \varphi, q') \in \Delta_A,\ \mathcal{C}_A(q, \varphi, q') = c \,\}.
        \]
        \item The coloring function $\mathcal{C}_B : Q_B \to [d]$ is given by
        \[
            \mathcal{C}_B((q,c)) = c \quad \text{for all } (q,c) \in Q_B.
        \]
    \end{enumerate}
    By construction, all transitions leaving a state $(q,c)$ in $\mathcal{B}$ receive the same color $c$. Hence, $\mathcal{B}$ has a state-based coloring function.

    Observe that each transition $(q,\varphi,q')$ of $\mathcal{A}$ with color $c$ is simulated in $\mathcal{B}$ by a transition from $(q,c)$ to $(q',c')$ for every $c' \in C(q')$. Thus, every run of $\mathcal{A}$ corresponds to a unique run of $\mathcal{B}$ visiting the same sequence of states in $Q_A$, extended with color annotations. Conversely, each run of $\mathcal{B}$ projects to a run of $\mathcal{A}$ with the same sequence of colors. Therefore, for any acceptance condition $\mathit{Acc}$ that depends on color sequences, $L_{\mathit{Acc}}(\mathcal{A}) = L_{\mathit{Acc}}(\mathcal{B})$.

      The size of $\mathcal{B}$ is quadratic in the size of $\mathcal{A}$ since the number of states in $\mathcal{B}$ is at most $|Q_A| \cdot d$, and the number of transitions in $\mathcal{B}$ is at most $|\Delta_A| \cdot d$ (hence, quadratic in those of $\mathcal{A}$), and no new formulas are introduced.
\end{proof}

\section{Decision problems for Hanoi Omega-Automata}

In this section, we study the \emph{non-emptiness} and \emph{inclusion} problems for Hanoi omega-automata.

\subsection{Non-emptiness of Hanoi Omega-Automaton} 
The non-emptiness problem asks whether $\mathcal{L}_{\Acc}(\mathcal{A}) \neq \emptyset$ for a given HOA $\mathcal{A}$ and accepting condition $\Acc$. We prove that this problem is \NP-complete for Emerson-Lei as well as other acceptance conditions given in \Cref{sec:HOA}. Membership in \NP \ follows from the fact that if $\mathcal{A}$ admits an accepting run, then it also admits a lasso-shaped accepting run of polynomially bounded size which preserves the set of infinitely and finitely occurring colors.

Before proving the small lasso property for HOA (\Cref{lemma:small_lasso}), we prove it for state-based HOA (\Cref{lemma:small_lasso_statebased} below).

\begin{lemma}\label{lemma:small_lasso_statebased}
    Let $\mathcal{A} = (Q, I, \ap, \Delta, \mathcal{C})$ be a state-based Hanoi omega-automaton and let $\rho$ be an accepting run of $\mathcal{A}$. Then there exist finite run segments $\rho_u$ and $\rho_v$ such that $|\rho_u|, |\rho_v| \leq |Q|^2$ (hence polynomial in $\size{\mathcal{A}}$), for which $\rho_u \rho_v^\omega$ is an accepting run of $\mathcal{A}$, and
    \[
        \Fin(\mathcal{C}(\rho)) = \Fin(\mathcal{C}(\rho_u \rho_v^\omega))
        \quad \text{and} \quad
        \Inf(\mathcal{C}(\rho)) = \Inf(\mathcal{C}(\rho_u \rho_v^\omega)).
    \]
\end{lemma}
\begin{proof}
   Consider a state-based Hanoi omega-automata $\mathcal{A}$,  an accepting run $\rho = q_1\varphi_1q_2\varphi_2\dots$ of $\cal {A}$. Let $S$ be the set of all infinitely occurring states in $\rho$. So, the set $\Inf(\mathcal{C}(\rho))$ is all the labels (colors) of the states in $S$ and $\mathit{elem}(\mathcal{C}(\rho))$ denotes the set of colors that occur in $\rho$. Then, $$\exists n \geq 1\,  \forall i\geq n, q_i \in S.$$ 
    
    Let $m> n$ be the smallest index such that $q_m = q_n$ and every color in $\Inf(\mathcal{C}(\rho))$ occurs at least once between $q_n$ and $q_m$ in $\rho$.

Let $\rho_{u'}$ be the finite run segment $\rho_{u'} = q_1 \varphi_1 \dots q_{n-1} \varphi_{n-1}$
and let $\rho_{v'} = q_n \varphi_n \dots q_{m-1} \varphi_{m-1}$.
Since all states occurring in $\rho_{v'}$ appear from position $n$ onward, i.e., $\forall q_i \in \rho_{v'}, i \geq n$, we have $\mathit{elem}(\mathcal{C}(\rho_{v'})) \subseteq \Inf(\mathcal{C}(\rho))$. Moreover, by the choice of $m$, every color in $\Inf(\mathcal{C}(\rho))$ appears at least once in $\rho_{v'}$, i.e., $\Inf(\mathcal{C}(\rho)) \subseteq \mathit{elem}(\mathcal{C}(\rho_{v'}))$. Hence,
\[
    \mathit{elem}(\mathcal{C}(\rho_{v'})) = \Inf(\mathcal{C}(\rho)).
\]

Consider the infinite run $\rho_{u'} \rho_{v'}^\omega$. This run is valid since $q_m = q_n$ in $\rho$, and therefore $\rho_{v'}$ forms a loop. Now we have a run that satisfies the lasso property. By construction of the run,
\[
    \Inf(\mathcal{C}(\rho_{u'} \rho_{v'}^\omega)) = \mathit{elem}(\mathcal{C}(\rho_{v'})) = \Inf(\mathcal{C}(\rho)).
\]

Furthermore, all states that occur only finitely often in $\rho$ appear in the run segment $q_1 \varphi_1 \dots q_{n-1}$. Thus, $\Fin(\mathcal{C}(\rho)) \subseteq \mathit{elem}(\mathcal{C}(\rho_{u'}))$. If there exists a color $c$ such that $c \in \mathit{elem}(\mathcal{C}(\rho_{u'})) \cap \Inf(\mathcal{C}(\rho))$, then, since $\Inf(\mathcal{C}(\rho)) = \mathit{elem}(\mathcal{C}(\rho_{v'}))$, it follows that $c \in \mathit{elem}(\mathcal{C}(\rho_{v'}))$. Therefore,
\[
    \mathit{elem}(\mathcal{C}(\rho_{u'})) \setminus \mathit{elem}(\mathcal{C}(\rho_{v'})) = \Fin(\mathcal{C}(\rho)).
\]

Hence, $\Fin(\mathcal{C}(\rho)) = \Fin(\mathcal{C}(\rho_{u'} \rho_{v'}^\omega))$. Note that $\rho_{u'}$ and $\rho_{v'}$ might not be of polynomial size.

    Now, to satisfy the \emph{small} lasso property, we have to argue that we can construct $\rho_u,\rho_v$ such that the run $\rho_u\rho_v^\omega$ will be accepting, $|\rho_u|$ and $|\rho_v|$ are polynomial in $\size{\mathcal{A}}$ with \[\Fin(\mathcal{C}(\rho_u\rho_v^\omega))=\Fin(\mathcal{C}(\rho_{u'}\rho_{v'}^\omega))=\Fin(\mathcal{C}(\rho)), \text{ and }\] \[\Inf(\mathcal{C}(\rho_u\rho_v^\omega))=\Inf(\mathcal{C}(\rho_{u'}\rho_{v'}^\omega))=\Inf(\mathcal{C}(\rho)).\] 
    Towards this, we need the following claim.
    \begin{claim}
        Given a state $q \in Q$ and a finite run segment $\rho$ of $\mathcal{A}$, there exists a run segment $\rho'$ such that
        \begin{enumerate}
            \item $\mathit{elem}(\mathcal{C}(\rho')) = \mathit{elem}(\mathcal{C}(\rho))$,
            \item the number of occurrences of $q$ in $\rho'$ is at most $|Q|$,
            \item $|\rho'| \leq |\rho|$
        \end{enumerate}
    \end{claim}

\begin{claimproof}
Let $q\in Q$ and let $\rho$ be a finite run segment of $\mathcal{A}$.  
Let the positions of $q$ in $\rho$ be $i_1,i_2, \ldots i_k$ where where $k$ is the number of occurrences of $q$ in $\rho$ and
\[
i_1 < i_2 < \dots < i_k.
\]
If $k \le |Q|$ there is nothing to prove.  So assume $k > |Q|$. Consider the $k-1$ consecutive subsegments of $\rho$ between these occurrences: let $\rho_t$ for $t \in [k-1]$, be the subsegment of $\rho$ from the occurrence of $i_t$ (included) to $i_{t+1}$ (excluded). Each $\rho_t$ is a run starting at $q$ and ending just before the next occurrence of $q$.

We show that there exists some index $t$ such that every state occurring in $\rho_t$ also occurs elsewhere in $\rho$ (outside $\rho_t$). Suppose, for contradiction, that every $\rho_t$ contains at least one state that does not appear anywhere else in $\rho$ (i.e., a state unique to $\rho_t$). Then the states that are unique to different $\rho_t$'s are pairwise distinct, so there are at least $k-1$ distinct states that are not $q$. Including $q$ itself, this yields at least $k$ distinct states appearing in $\rho$. But since $k>|Q|$, it contradicts that $\rho$ can only visit states from the finite set $Q$. 

Therefore there exists some index $t$ such that every state occurring in $\rho_t$ also occurs elsewhere in $\rho$ (outside $\rho_t$). Because $\mathcal{A}$ is state-based, colors are determined by states. Hence removing the entire subsegment $\rho_t$ does not remove any color from the set of colors seen along $\rho$. Let $\rho'$ be the run obtained from $\rho$ by deleting the subsegment $\rho_t$. Then:
\begin{itemize}
    \item $|\rho'| < |\rho|$ (we removed at least one transition),
    \item $\mathit{elem}(\mathcal{C}(\rho')) = \mathit{elem}(\mathcal{C}(\rho))$ (no color disappeared),
    \item the number of occurrences of $q$ in $\rho'$ is $k-1$.
\end{itemize}

If $k-1 > |Q|$, repeat the argument on $\rho'$. Each repetition reduces the number of occurrences of $q$ by at least one while preserving the colors. After finitely many steps we obtain a run segment ($\rho'$) with at most $|Q|$ occurrences of $q$, with the desired properties.
\end{claimproof}
    
By iteratively using the claim above for all states $q$ occuring in the runs $\rho_{u'}, \rho_{v'}$, we can construct $\rho_u$ from $\rho_{u'}$ and $\rho_{v}$ from $\rho_{v'}$. After these reductions, every state appears at most $|Q|$ times in $\rho_u$ and $\rho_v$, and since there are at most $|Q|$ distinct states, the total number of state occurrences in $\rho_u$ and $\rho_v$ is at most $|Q|^2$. Hence, $|\rho_u|\leq |Q|^2$ and $|\rho_v|\leq |Q|^2$. Since the set of colors is preserved between $\rho_v$ and $\rho_{v'}$, and $\rho_u$ and $\rho_{u'}$, we have: $\Inf(\mathcal{C}(\rho))=\Inf(\mathcal{C}(\rho_u\rho_v^\omega))$, $\Fin(\mathcal{C}(\rho))=\Fin(\mathcal{C}(\rho_u\rho_v^\omega))$ and $|\rho_u|,|\rho_v| \leq |Q|^2$.
\end{proof}

\begin{lemma}\label{lemma:small_lasso}
Given an Hanoi omega-automaton $\mathcal{A} = (Q, I, \ap, \Delta, F)$ with index $d$ and an accepting run $\rho$, there exist $\rho_u,\rho_v$ such that $|\rho_u| ,|\rho_v| \leq (|Q| \cdot d)^2$, $\rho_u\rho_v^\omega$ is an accepting run, $\Fin(\mathcal{C}(\rho))=\Fin(\mathcal{C}(\rho_u\rho_v^\omega))$ and $\Inf(\mathcal{C}(\rho)) = \Inf(\mathcal{C}(\rho_u\rho_v^\omega))$.
\end{lemma}

\begin{proof}
    The idea is to first convert the HOA into an equivalent state-based automaton, which increases the number of states to $|Q| \cdot d$ (see \Cref{lem:to_state_based}). Then, by applying \Cref{lemma:small_lasso_statebased} to an accepting run $\rho$, we obtain a lasso-shaped run $\rho_u \rho_v^\omega$ with $|\rho_u|, |\rho_v| \le (|Q| \cdot d)^2$, which preserves both the finitely and infinitely occurring colors.
\end{proof}

\Cref{lemma:small_lasso} allows one to derive an \NP~upper-bound for the non-emptiness problem of HOA automata with Emerson-Lei acceptance: it suffices to non-deterministically guess a lasso run of polynomial length and a valuation for each of its transitions, and then check that the valuations satisfy the formulas on transitions, and that the sequence of colours seen finitely and infinitely often in the run satisfy the acceptance formula. The \NP~lower bound, which holds already for safety and reachability HOA, is a rather direct reduction of SAT, as HOA transitions are labeled with Boolean formulas. This yields the following theorem.

\begin{theorem}\label{thm:non-emptiness}
The non-emptiness problem for Hanoi omega-automaton with an Emerson-Lei acceptance condition is \NP-complete. It is \NP-hard already for deterministic safety and reachability HOA automata.
\end{theorem}

\begin{proof}

  We first prove that non-emptiness for HOA with Emerson–Lei acceptance is in \NP, and then show \NP-hardness already for deterministic HOA with safety or reachability acceptance conditions. Consequently, the non-emptiness problem for HOA with Emerson–Lei acceptance is \NP-complete.

 Consider an HOA $\mathcal{A}$ with index $d$ and Emerson–Lei acceptance, and suppose $\mathcal{L}(\mathcal{A})\neq \emptyset$. By \Cref{lemma:small_lasso} there exists an accepting run of the form $\rho_u \rho_v^\omega$ where the finite segments
   $\rho_u$ and $\rho_v$ satisfy $|\rho_u|,|\rho_v| \leq (|Q|\cdot d)^2$ (indeed polynomial in $\size{\mathcal{A}}$). Thus a polynomial size certificate for non-emptiness is the pair of run segments $\rho_u,\rho_v$ together with valuation segments $w_u, w_v$, which witness the satisfiability of all Boolean formulas occurring along $\rho_u$ and $\rho_v$. Given such a certificate, acceptance can be verified in polynomial time as follows:
   \begin{enumerate}
     \item Verify that $\rho_u\rho_v^\omega$ is a valid run of $\mathcal{A}$ and that $w_u \models \rho_u$ and $w_v \models \rho_v$.
       \item Simulate $\mathcal{A}$ on $\rho_u \rho_v^\omega$
   by simulating the finite prefix $\rho_u$ and one period of $\rho_v$ (both polynomial length) to compute the sets $\Fin(\mathcal{C}(\rho_u \rho_v^\omega))$ and $\Inf(\mathcal{C}(\rho_u \rho_v^\omega))$ and then evaluate the Emerson–Lei Boolean formula on those sets.
   \end{enumerate}
   Verification of the run, simulation, and evaluation of the acceptance condition all take time polynomial in $\size{\mathcal{A}}, |\rho_u|$ and $|\rho_v|$, so non-emptiness belongs to \NP.
   
   It remains to show \NP-hardness. We reduce from SAT. Let $\varphi$ be a Boolean formula over atomic propositions $\mathcal{P}$. In polynomial time, we construct a deterministic HOA automaton $\mathcal{A}$ (shown below) with reachability acceptance such that $\mathcal{L}(\mathcal{A}) \neq \emptyset \iff \varphi$ is satisfiable. 

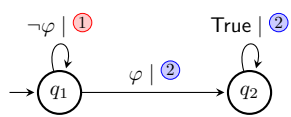
\begin{center}
    \begin{tikzpicture}[shorten >=1pt, auto, node distance=2.5cm, on grid, >=stealth, scale=0.8, every node/.style={scale=0.8}]
    \tikzset{state/.style={circle, draw=black, thick, minimum size=0.7cm}}
    
    % Nodes
    \node[state, initial, initial text=] (q1) {$q_1$};
    \node[state] (q2) [right=of q1] {$q_2$};
    
    % Transition labels with colored numbers in circles
    \path[->]
        (q1) edge [above] node {$\varphi \mid \tikz[baseline]{\node[circle, draw=blue,fill=blue!20, inner sep=1pt]{2};}$} (q2)
        (q1) edge [loop above] node {$\neg \varphi \mid \tikz[baseline]{\node[circle, draw=red, fill=red!20,inner sep=1pt]{1};}$} (q1)
        (q2) edge [loop above] node {$\mathsf{True} \mid \tikz[baseline]{\node[circle, draw=blue, fill=blue!20, inner sep=1pt]{2};}$} (q2);
\end{tikzpicture}
\end{center}
The reachability acceptance is given by $R = \{2\}$, i.e., it requires that color~2 occurs at least once during the run. Therefore, the automaton admits a run if and only if there exists a valuation satisfying $\varphi$, that is, if $\varphi$ is satisfiable. Similarly, if $\mathcal{A}$ is equipped with the safety acceptance condition $S = \{2\}$, a run is accepting only if all transitions occurring in the run are labeled with color~2. Again, the automaton admits an accepting run if and only if $\varphi$ is satisfiable.
\end{proof}

\begin{remark}
    As a consequence of \Cref{thm:non-emptiness}, the non-emptiness problem is \NP-complete for HOA with Muller, Streett, Rabin, parity, B\"uchi, co-B\"uchi, reachability and safety acceptance.
\end{remark}

\subsection{Language Inclusion of Hanoi Omega-Automaton} 
Given two HOAs $\cal A$ and $\cal B$ equipped with acceptance conditions, and over the same set of atomic propositions $\mathcal{P}$, the \emph{language inclusion problem} asks whether every infinite word accepted by $\cal A$ is also accepted by $\cal B$. We show that this problem is \PSPACE-complete for all acceptance conditions discussed in \Cref{sec:HOA}, except for Emerson-Lei, where the problem becomes \EXPSPACE-complete. We start with acceptance conditions up to Streett.

\begin{theorem}\label{thm:inclusionStreett}
    If $\mathcal{A}$ and $\mathcal{B}$ are Hanoi omega-automata with Streett, parity, Buchi, co-Buchi, reachability or safety accepting conditions, then determining whether $\lan{A} \subseteq \lan{B}$ is \PSPACE-complete.
\end{theorem}

\begin{proof}
    The \PSPACE-hardness follows directly from the known \PSPACE-hardness of language inclusion for finite automata \cite{Stockmeyer1973}. For the upper bound, consider two HOAs $\mathcal{A}$ and $\mathcal{B}$ equipped with Streett acceptance. By virtue of \Cref{lem:to_state_based}, we can assume that $\mathcal{A}$ and $\mathcal{B}$ are state-based HOAs. We can translate each of them into an equivalent classical Streett automaton $\mathcal{A}'$ and $\mathcal{B}'$ over the alphabet $\val(\mathcal{P})$, by replacing every transition labeled with a Boolean formula $\varphi$ by all transitions labeled with valuations satisfying $\varphi$. This construction preserves the set of states unchanged but may introduce an exponential number of transitions. Hence, explicitly constructing $\mathcal{A}'$ and $\mathcal{B}'$ would in general require exponential time and space, which we avoid by simulating them on the fly. 
    
    To decide $\mathcal{L}(\mathcal{A}) \subseteq \mathcal{L}(\mathcal{B})$, it suffices to check emptiness of $\mathcal{L}(\mathcal{A}') \cap \overline{\mathcal{L}(\mathcal{B}')}$. To handle complementation, using Piterman's improved versions of Safra's determinisation construction \cite{Pit06}, the non-deterministic Streett automaton $\mathcal B'$ with $n$ states and $k$ pairs (acceptance size) can be determinised into an equivalent deterministic parity automaton $\mathcal D_{B'}$ with $2^{\Theta(nk\log(nk))}$ states and $\Theta(nk)$ pairs. The states of this parity automaton are Safra trees, and each transition of $\mathcal D_{B'}$ is computable in polynomial space i.e., given a Safra tree $T$ and a valuation $v$, the next state, denoted by $\mathrm{Safra}(T,v)$, can be computed in polynomial space. Complementation is then obtained by dualising the parity condition of $\mathcal{D}_{\mathcal{B}'}$, yielding an automaton recognising $\overline{\mathcal{B}'}$ without increasing the state space.

For checking the emptiness of $\mathcal L(\mathcal A') \cap \mathcal L(\overline{\mathcal B'})$, we then take the intersection of Streett automata $\mathcal{A}'$ and $\overline{\mathcal B'}$ by their synchronous product construction to get automaton $\mathcal C$, whose state space has size 
$\size{\mathcal A'} \cdot 2^{\Theta(nk\log(nk))} = \lassowordlimit$. If $\mathcal L(\mathcal C)\neq\emptyset$, then by \Cref{lemma:small_lasso} there exists an accepting lasso of exponential length.

 We now describe a non-deterministic polynomial space algorithm to check emptiness of $\mathcal L(\mathcal{A}) \cap \mathcal L(\overline{\mathcal{B}})$. The algorithm guesses a lasso-shaped run of length at most $\lassowordlimit$ and simulates transitions of $\mathcal A'$ and $\overline{\mathcal B'}$ on the fly (without constructing them explicitly) as follows:  At each step, the algorithm non-deterministically choose a valuation $v$ and selects transitions of $\mathcal{A}$ and $\mathcal{B}$ whose Boolean formulas are satisfied by $v$, thereby simulating $\mathcal{A}'$ and $\mathcal{B}'$ on the fly. Initially, let $q^{A'}$ be the initial state of $\mathcal{A}'$ and $T^{B'}$ be the initial safra tree of $\mathcal{B}'$. Given the current Safra tree $T^{B'}$, the next tree $\mathit{Safra}(T^{B'},v)$ is computed from the chosen valuation $v$ using a single step of Piterman’s determinisation construction. This step can be performed using polynomial space, and at any point the algorithm stores only the current state $q^{A'}$ and the current Safra tree $T^{B'}$. To detect an accepting lasso, the algorithm guesses a position $p$ corresponding to the start of the cycle and stores the configuration $(q^{A'}_p, T^{B'}_p)$ at that point. It then continues the simulation and checks whether the same configuration is reached again. During this phase, it maintains two pieces of information: (i) the maximal color $c$ seen along the cycle (for the parity condition corresponding to $\overline{\mathcal{B}'}$), and (ii) the set $S$ of visited states (from $\mathcal{A}$) to verify the Streett condition of $\mathcal{A}$. If the configuration $(q^{A'}_p, T^{B'}_p)$ is revisited and both acceptance conditions are satisfied (i.e., $c$ is odd and $S$ satisfies the Streett acceptance of $\mathcal{A}$), the algorithm accepts. Although the lasso may have exponential length, the algorithm never stores the entire run. It only keeps the current configuration, the guessed cycle entry configuration, and the values $c$ and $S$, all of which require polynomial space. If the intersection is non-empty, such a lasso exists within the guessed bound.

Therefore, the algorithm uses only polynomial space. Since non-deterministic polynomial space coincides with \PSPACE \ (by Savitch’s Theorem), emptiness of $\mathcal L(\mathcal{A}) \cap \mathcal L(\overline{\mathcal{B}})$ can be decided in \PSPACE.

Finally, acceptance conditions such as parity, Büchi, co-Büchi, reachability, and safety are special cases of Streett conditions. Therefore, their language inclusion problems also lie in \PSPACE.
\end{proof}

In the case of Rabin and Muller HOAs, we show that each automaton can be transformed into an equivalent Streett HOA of polynomial size. The construction follows the classical Rabin/Muller to Streett conversions \cite{Boker17, Choueka74}. For the sake of completeness, we provide the proof below. Using \Cref{thm:inclusionStreett}, this implies \PSPACE-completeness of the language inclusion problem for both acceptance conditions.

\begin{lemma}\label{lem:HOARabinMuller_to_Streett}
    Every Rabin or Muller HOA can be effectively transformed into an equivalent Streett HOA with only a polynomial blow-up in size.
\end{lemma}

\begin{proof}
    We consider the Rabin and Muller cases separately, as the constructions are different. Both proofs follow the classical Rabin/Muller to Streett conversions \cite{Boker17, Choueka74}. By virtue of \Cref{lem:to_state_based}, we may restrict attention to state-based Rabin and Muller HOAs.

    \begin{description}
        \item[Muller to Streett.] A Muller HOA with $n$ states and $k$ Muller sets can be viewed as a union of $k$ Muller HOAs, each with $n$ states (same structure) and 1 Muller set. We show that each Muller HOA with a single Muller set can be converted into a Streett HOA of polynomial size. For the general Muller HOA, taking a disjoint union of all the corresponding Streett HOAs yields a Streett HOA that is equivalent to the original Muller HOA and is of polynomial size.

        Let $\mathcal{A} = (Q, Q_0, \mathbb{B}(\ap), \Delta, \mathcal{C})$ be a state-based Muller HOA with index $d$ and single Muller set $E= \{c_1, c_2, \ldots, c_l\}$ consisting of at most $l \leq d \leq |Q|$ colors that must be visited infinitely often. We construct an equivalent Streett HOA $\cal B$ with acceptance condition $S$ as follows: $\cal B$ consist of two copies of  $\cal A$, where the second one is a restricted copy of $\cal A$ containing only those states whose assigned color belongs to $E$, i.e., $\{q \in Q \mid {\cal C} (q) \in E\}$. The transitions in the first copy are as in $\cal A$, and in addition allow to move to the corresponding states in the second copy. The Streett acceptance condition will require all the colors of the second copy to be visited infinitely often, and is given by, $S= \{(\emptyset,c_1), (c_1, c_2) , \ldots (c_{l-1},c_l)\}$. As a result, $\cal B$ has at most $2|Q|$ states and $|Q|$ accepting pairs. Hence, $\size{\cal B}$ is polynomial w.r.t. $\size{\cal A}$.

        \item[Rabin to Streett.] We detail the construction from Rabin to B\"uchi HOA (which is a special case of Streett HOA). A Rabin HOA with $n$ states and $k$ Rabin pairs can be seen as a union of $k$ Rabin HOA, each with $n$ states (same structure) and a single Rabin pair. We show that each Rabin HOA with single Rabin pair can be converted to an equivalent B\"uchi HOA of polynomial size. Consequently, taking the disjoint union of these B\"uchi HOAs yields a B\"uchi HOA which is equivalent to the original Rabin HOA, and the resulting automaton remains of polynomial size.

        Let $\cal {A}$ be a Rabin HOA with a single Rabin pair $(E,F)$, where colors in $E$ must be visited infinitely often and colors in $F$ must be visited only finitely often. The equivalent B\"uchi HOA $\cal {B}$ is constructed as follows: $\cal {B}$ contains two copies of $\cal {A}$, where the second copy includes only those states whose colors do not belong to $F$. The first copy non-deterministically guesses a point to jump to the second copy, thereby ensuring that colors in $F$ are avoided from that point onward. The second copy has no transitions back to the first copy and uses $E$ as its B\"uchi acceptance set. Hence,$\size{\cal B}$ is polynomial in $\size{\cal A}$.

    \end{description}
\end{proof}

Using \Cref{lem:HOARabinMuller_to_Streett} and \Cref{thm:inclusionStreett}, we obtain the following result
\begin{theorem}
     The language inclusion problem for Muller and Rabin HOAs is \PSPACE-complete.
\end{theorem}

Next, it remains to study the problem for Emerson-Lei HOA. In this case, we show that every Emerson-Lei HOA can be effectively transformed into B\"uchi HOA with an exponential blow-up in size.  The proof is adapted from the standard Emerson–Lei to Büchi construction (see \cite{DBLP:conf/stoc/SafraV89}), and for completeness sake, we provide the proof in the below.
\begin{lemma}\label{lem:HOA-EL_toHOA-B}
    Every HOA $\mathcal{A}$ with Emerson-Lei acceptance can be effectively transformed to an equivalent HOA with B\"uchi acceptance of size $2^{\mathcal{O}(\size{\mathcal{A}})}$.
\end{lemma}
\begin{proof}
    We adapt the standard Emerson-Lei to Büchi construction \cite[Proposition 3.1]{DBLP:conf/stoc/SafraV89}. Let $\mathcal{A}$ be a HOA automaton with Emerson-Lei acceptance formula $\alpha$. Let $\alpha'$ denote the disjunctive normal form of \(\alpha\), i.e., $\alpha' \;=\; \bigvee_{i=1}^{k} \alpha_i$, where $k \le 2^{\size{\alpha}}$ and each $\alpha_i$ is a conjunction of predicates of the form $\mathsf{Inf}(C)$ or $\mathsf{Fin}(C)$. 
    
    For each disjunct $\alpha_i$, we construct a HOA $\mathcal{A}_i$ with B\"uchi acceptance that enforces $\alpha_i$.  WLOG let $\alpha_i = \bigwedge_{j=1}^{m} \mathsf{Inf}(E_j) \land \bigwedge_{r=1}^{n} \mathsf{Fin}(F_r)$ for some $m,n \geq 0$.  The conjunct $\bigwedge_{j=1}^{m} \mathsf{Inf}(E_j)$ is a generalized-Büchi condition:  each set $E_j$ must be visited infinitely often. We convert this generalized B\"uchi condition into a single B\"uchi condition using a standard “visit all Inf-sets cyclically” by introducing a round-robin counter over $\{1,\dots,m\}$ that cycles through the $m$ $\mathsf{Inf}$ predicates. The conjunct $\bigwedge_{r=1}^{n} \mathsf{Fin}(F_r)$ is equivalent to a single co-Büchi condition $\mathsf{Fin}(C)$ with $C \;=\; \bigcup_{r=1}^{n} F_r$, meaning that the run may visit $C$ only finitely often. We implement this by introducing a flag bit $b \in \{0,1\}$: at some point the automaton non-deterministically guesses the last visit to $C$ and switches $b$ from $0$ to $1$.  
    In mode $b=1$, all transitions labeled by colors in $C$ are forbidden. 
    Thus, the HOA B\"uchi automata ${\cal A}_i$ is obtained as a synchronous product of $\cal A$ with a round-robin counter $\{1, \ldots, m\}$ for the $\mathsf{Inf}$ predicates and a flag bit $\{0,1\}$ for the $\mathsf{Fin}$ predicates. The size of ${\cal A}_i$ is at most $2 \cdot \size{\alpha} \cdot \size{\cal A}$ since $m,n \leq \size{\alpha}$.  For acceptance, we assign a new color $c$ to exactly those transitions where the round-robin counter returns to $1$ in mode $b=1$.  Thus the B\"uchi acceptance condition is $\mathsf{Inf}(\{c\})$. This ensures that each $E_j$ is visited infinitely often and $C$ is visited only finitely often (flag remains in mode $b=1$).

     Since $\alpha'$ is a disjunction, we have $\mathcal{L}(\mathcal{A}) = \bigcup_{i=1}^{k} \mathcal{L}(\mathcal{A}_i)$. We obtain an equivalent B\"uchi automaton by taking the disjoint union of all $\mathcal{A}_i$.  The size of the resulting automaton is bounded by $2^{\size{\mathcal{A}}}$, since $k \le 2^{\size{\alpha}}$ and each $\mathcal{A}_i$ has size polynomial in $\size{\mathcal{A}}$. Thus $\mathcal{A}$ can be translated into an equivalent B\"uchi HOA automaton of size at most $2^{\mathcal{O}(\size{\mathcal{A}})}$.
\end{proof}

\begin{remark}
    Since a B\"uchi HOA  is in particular also a parity/Streett/Rabin HOA, \Cref{lem:HOA-EL_toHOA-B} also gives an exponential translation from HOA Emerson-Lei to HOA parity/Streett/Rabin. 
\end{remark}

\begin{theorem}\label{thm:inclusionEL}
   The language inclusion problem for Emerson-Lei HOA is \EXPSPACE-complete.
\end{theorem}
\begin{proof}
    The lower bound follows from \cite[Proposition 4.6]{DBLP:conf/stoc/SafraV89}, where it is shown that the universality problem for Emerson–Lei automata is \EXPSPACE-complete.
For the upper bound, given two HOA automata with Emerson-Lei acceptance, we convert them into HOA automata with Büchi acceptance using \Cref{lem:HOA-EL_toHOA-B}. This transformation causes an exponential blow-up in the number of states. Since language inclusion for HOA with Büchi acceptance is \PSPACE-complete (\Cref{thm:inclusionStreett}), the exponential increase in the state space yields an overall \EXPSPACE \ upper bound.
\end{proof}

\section{Two-Player Hanoi Omega-Games}\label{sec:hog}
A \emph{two-player Hanoi omega-game} (HOG) is a tuple 
\[{\cal G} = ({\cal A},\ap_\inp,\ap_\outp,Acc)\]
such that ${\cal A}=(Q,q_0,{\cal P},\delta, \cal C)$ is a complete\footnote{For every state $q$ and every valuation $v: \ap\rightarrow \{\mathsf{True},\mathsf{False}\}$, there exists a transition from $q$ on $v$.} deterministic HOA automaton over the atomic propositions $\ap = \ap_\inp \uplus \ap_\outp$ with index $d$, and $Acc\subseteq [d]^\omega$ is a called a winning condition.

Intuitively, a HOG is played in turns by two players, Player \textsf{In} and Player \textsf{Out}, who successively pick Boolean valuations of the atomic propositions in $\ap_\inp$ and $\ap_\outp$ respectively. Player \textsf{In} plays first, picking some valuation $v_1^\inp : \ap_\inp\rightarrow \{\mathsf{True},\mathsf{False}\}$, then Player \textsf{Out} picks some valuation $v_1^\outp : \ap_\outp\rightarrow \{\mathsf{True},\mathsf{False}\}$, and the game continues like that \emph{ad infinitum}, forming an infinite sequence $\pi = (v_1^\inp\cup v_1^\outp)(v_2^\inp\cup v_2^\outp)\dots\in (\val(\ap))^\omega$. Player \textsf{Out} wins if $\pi$ belongs to $\mathcal{L}_{Acc}({\cal A})$, otherwise Player \textsf{In} wins. The game is then won by Player \textsf{Out} if he wins whatever the valuations picked by Player \textsf{In}, in other words, if he has a winning strategy. 

Formally, a \emph{play} $\pi$ is an infinite sequence of valuations, i.e., $\pi\in (\val(\ap))^\omega$.  
A \emph{strategy} for Player \textsf{Out} is a mapping $\lambda_\outp : (\val(\ap))^* \cdot \val(\ap_\inp)\rightarrow \val(\ap_\outp)$, while a strategy for Player \textsf{In} is a mapping $\lambda_\inp : (\val(\ap))^*\rightarrow \val(\ap_\inp)$. We denote by $\Lambda_{\outp}$ and $\Lambda_{\inp}$ the set of strategies of Player \textsf{Out} and Player \textsf{In} respectively. The \emph{play induced by a pair of strategies} $(\lambda_\inp,\lambda_\outp)\in \Lambda_\inp\times \Lambda_\outp$ is the play
$
\textsf{play}(\lambda_\inp,\lambda_\outp) = (v_1^\inp\cup v_1^\outp)(v_2^\inp\cup v_2^\outp)\dots
$
such that for all $i\geq 1$, it holds that
$v_i^\inp = \lambda_\inp((v_1^\inp\cup v_1^\outp)\dots (v_{i-1}^\inp\cup v_{i-1}^\outp))$ and $v_i^\outp = \lambda_\outp((v_1^\inp\cup v_1^\outp)\dots (v_{i-1}^\inp\cup v_{i-1}^\outp)v_i^\inp)$.

The \emph{language of a strategy} $\lambda_\outp\in\Lambda_\outp$ is the set $\mathcal{L}(\lambda_\outp) = \{ \textsf{play}(\lambda_\inp,\lambda_\outp)\mid \lambda_\inp\in\Lambda_\inp\}$. The language of a strategy $\lambda_\inp\in\Lambda_\inp$ is defined in a similar way. A strategy $\lambda_\outp\in\Lambda_\outp$ is winning in $\mathcal{G}$ if all plays in $\mathcal{L}(\lambda_\outp)$ are winning, i.e. $\mathcal{L}(\lambda_\outp)\subseteq \mathcal{L}_{Acc}({\cal A})$. A strategy $\lambda_\inp\in\Lambda_\inp$ is winning in $\mathcal{G}$ if $\mathcal{L}(\lambda_\inp)\subseteq \overline{\mathcal{L}_{Acc}({\cal A})}$. Player \textsf{Out} (resp. Player \textsf{In}) wins an HOG ${\cal G}$ (with winning condition $Acc$) if there exists a winning strategy for Player \textsf{Out} (resp. Player \textsf{In}). A class of HOGs is \emph{determined} if for all HOG in the class, either Player \textsf{In} has a winning strategy or Player \textsf{Out} has a winning strategy.

\begin{proposition}\label{prop:HOG_determinacy}
    HOGs are determined for $\omega$-regular winning conditions\footnote{A winning condition $Acc$ is called $\omega$-regular if it
can, for instance, be expressed by some Emerson-Lei acceptance formula (See \Cref{sec:HOA}).}.
\end{proposition}

\begin{proof}
This follows from (Borel) determinacy of classical games \cite{martin1975borel} because
one can make the transitions of HOGs explicit and construct a classical $\omega$-regular games, with strategies being preserved in both directions. For the construction, the transitions are made explicit by replacing every transition labeled with a Boolean formula $\varphi$ by all transitions labeled with valuations satisfying $\varphi$. 
\end{proof}

For a complexity class ${\cal C}$, we say that an HOG ${\cal G}$ together with a winning condition $Acc$ is in ${\cal C}$ (resp. is ${\cal C}$-complete) if the problem of checking whether Player \textsf{Out} has an $Acc$-winning strategy in ${\cal G}$ is in ${\cal C}$ (resp. is ${\cal C}$-complete).

We show that Hanoi omega-games with Streett, parity, B\"uchi, co-B\"uchi, reachability and safety winning conditions are $\boldsymbol{\Pi_2}$-complete.  
We start with the lower bound, which already holds for HOG with reachability and safety winning conditions.

\begin{lemma}\label{lemma:HOGpi_2hard}
Reachability and safety HOGs (and hence Rabin/Streett/ parity/ B\"uchi/ co-B\"uchi HOGs) are $\boldsymbol{\Pi_2}$-hard.
\end{lemma}

\begin{proof}
    We reduce from the $\boldsymbol{\Pi_2}$-complete problem $\textsc{QSAT}^\forall_2$.
   Given a formula of the form $\forall \overline{x}\exists \overline{y}\varphi$ where $\varphi$ is quantifier-free, we construct an HOG $G_\varphi = ({\cal A}, \overline{x},\overline{y}, Acc)$ where ${\cal A} = (\{q_0,q_1,q_2\}, q_0, {\cal P}, \delta, \cal C)$ is the HOA depicted below, and the proposition set $\mathcal{P}$ consists of all variables in \(\overline{x}\) and \(\overline{y}\).  The reachability winning condition $Acc$ is given by the set $R=\{2\}$, i.e, it requires that color $2$ occurs at least once during the run. Player \textsf{Out} wins for the reachability winning condition iff it can always force the play to eventually reach \(q_2\), regardless of Player~\textsf{In}'s choice of $\overline{x}$.  This occurs exactly when $\forall \overline{x}\exists \overline{y}\varphi$ is true. Hence, Player \textsf{Out} wins iff $\forall \overline{x}\exists \overline{y}\varphi$ is satisfiable. Thus,  reachability HOGs are $\boldsymbol{\Pi_2}$-hard.
   
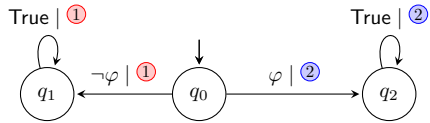
\begin{center}
    \begin{tikzpicture}[shorten >=1pt, auto, node distance=2.5cm, on grid, >=stealth, scale=0.8, every node/.style={scale=0.8}]

  % States
  \node[state, initial, initial above, initial text=] (q0) {$q_0$};
  \node[state, right=of q0] (q1) {$q_2$};
  \node[state, left of=q0] (q2) {$q_1$};

  \path[->]
    (q0) edge[left] node[above] {$\varphi \mid \tikz[baseline]{\node[circle, draw=blue,fill=blue!20, inner sep=1pt]{2};}$} (q1)
    (q0) edge[right] node[above] {$\lnot \varphi \mid \tikz[baseline]{\node[circle, draw=red,fill=red!20, inner sep=1pt]{1};}$} (q2)
    (q1) edge[loop above] node {$\mathsf{True}  \mid \tikz[baseline]{\node[circle, draw=blue,fill=blue!20, inner sep=1pt]{2};}$} ()
    (q2) edge[loop above] node {$\mathsf{True} \mid \tikz[baseline]{\node[circle, draw=red,fill=red!20, inner sep=1pt]{1};}$} ();
\end{tikzpicture}
\end{center}

Similarly, if $\cal A$ is equipped with the safety winning condition given by $S = \{2\}$, so that only transitions labeled with color 2 are allowed, Player \textsf{Out} wins iff it can always take the transition leading to \(q_2\). This is possible exactly when
\(\forall \overline{x}\, \exists \overline{y}\, \varphi\) is true. Therefore, safety HOGs are also $\boldsymbol{\Pi_2}$-hard.
\end{proof}

\subparagraph*{HOG to Classical Game.} Towards showing the upper bound for solving HOGs with $\omega$-regular winning conditions, we define a classical game $P_{\cal G}$ for a given HOG ${\cal G} = ({\cal A},\ap_\inp,\ap_\outp, Acc)$. WLOG, assume that ${\cal A} = (Q, q_0,\ap_\inp \cup \ap_\outp, \delta, \cal C)$ is a state-based HOA with index $d$. Since $\mathcal{A}$ is deterministic, for any valuation $v \in \val(\ap_\inp \cup \ap_\outp)$, we write $\delta(q,v)$ to denote the (unique) state $p$ such that $\delta(q,\varphi)=p$ and $v \models \varphi$ (if such $p$ exists).  For a state $q \in Q$, a subset $S \subseteq Q$ is called \emph{$q$-good} if there exists a valuation $v_\inp \in \val(\ap_\inp)$ such that for all $v_\outp \in \val(\ap_\outp)$, we have $\delta(q, v_\inp \cup v_\outp) \in S$. Intuitively, $S$ is $q$-good if Player $\textsf{In}$ can choose a valuation guaranteeing that the next state always lands in $S$, no matter what Player $\textsf{Out}$ does.  

    We now construct the classical turn-based game $P_{\cal G}$ as follows: \[P_{\cal G} = (V_\inp \cup V_\outp, V_\inp=Q, V_\outp=Q \times 2^Q,E,C, \mathcal{C}', Acc')\]  where the edge relation $E$ contains  i) an edge from $q \in V_\inp$ to $(q,S) \in V_\outp$ if $S$ is $q$-good ii) an edge from $(q,S) \in V_\outp$ to $p \in V_\inp$ if $p \in S$. The set of colors $C = [d] \cup \{0\}$ where \(0 \notin [d]\) is a fresh color. The colouring $\mathcal{C}' : V_\inp \cup V_\outp \to [d] \cup \{0\}$ is given by ${\cal C'}(q)= {\cal C}(q)$ for $q \in V_\inp$ and ${\cal C'}(q,S) = 0$ for $(q,S) \in V_\outp$.
     The winning condition $Acc'$ for $P_{\cal G}$ is defined as \[Acc' = \{ c_1 0 c_2 0 \ldots \mid c_1 c_2 \ldots \in Acc\}.\]
    The HOA given in \Cref{fig:HOG} has a state-based coloring function and hence, can be viewed as a state-based HOA. The game involving this HOA is given by ${\cal G} = ({\cal A}, \ap_{\inp} = \{r\}, \ap_{\outp} = \{g\}, Acc = \textsf{Inf}(\{1\}))$ and the classical game $P_{\cal G}$ is given below. In this instance, one can verify that, for example, the sets $\{q_1, q_2\}$ and $\{q_2, q_3\}$ are $q_1$-good, witnessed by the $\ap_{\inp}$ valuations $r \mapsto \textsf{False}$ and $r \mapsto \textsf{True}$, respectively.
     \begin{center}
         \scalebox{0.75}
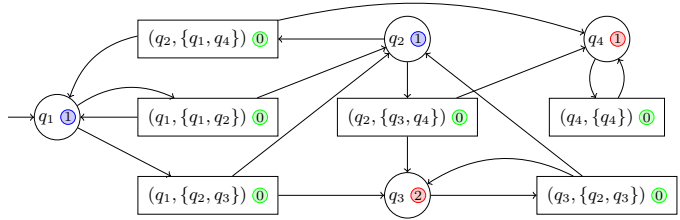
{
         \begin{tikzpicture}[
    node distance=3.5cm and 2cm,
    every node/.style={font=\small},
    circ/.style={circle, draw, minimum size=8mm},
    rect/.style={rectangle, draw, minimum width=12mm, minimum height=6mm}
]

% Central rectangle
\node[rect,minimum size=2mm, inner sep=5pt] (b) {$(q_1, \{q_1,q_2\})\;\tikz[baseline=(n.base)]{
\node[circle, draw=green, fill=green!20, inner sep=0.5pt, font=\scriptsize] (n) {0};
}$};

% Vertical rectangles
\node[rect, above=0.7cm of b,minimum size=2mm, inner sep=5pt] (c) {$(q_2, \{q_1,q_4\})\;\tikz[baseline=(n.base)]{
\node[circle, draw=green, fill=green!20, inner sep=0.5pt, font=\scriptsize] (n) {0};
}$};
\node[rect, below=0.7cm of b,minimum size=2mm, inner sep=5pt] (a) {$(q_1, \{q_2,q_3\})\;\tikz[baseline=(n.base)]{
\node[circle, draw=green, fill=green!20, inner sep=0.5pt, font=\scriptsize] (n) {0};
}$};

% Horizontal rectangles
\node[rect, right of=b,minimum size=2mm, inner sep=5pt] (d) {$(q_2, \{q_3, q_4\})\;\tikz[baseline=(n.base)]{
\node[circle, draw=green, fill=green!20, inner sep=0.5pt, font=\scriptsize] (n) {0};
}$};
\node[rect, right of=d,minimum size=2mm, inner sep=5pt] (f) {$(q_4, \{q_4\})\;\tikz[baseline=(n.base)]{
\node[circle, draw=green, fill=green!20, inner sep=0.5pt, font=\scriptsize] (n) {0};
}$};

% Circles
\node[circ, right of=c,minimum size=2mm, inner sep=1pt] (z) {$q_2\;\tikz[baseline=(n.base)]{
\node[circle, draw=blue, fill=blue!20, inner sep=0.5pt, font=\scriptsize] (n) {1};
}$};
\node[circ, right of=z,minimum size=2mm, inner sep=1pt] (w) {$q_4\;\tikz[baseline=(n.base)]{
\node[circle, draw=red, fill=red!20, inner sep=0.5pt, font=\scriptsize] (n) {1};
}$};
\node[circ, right of=a,minimum size=2mm, inner sep=1pt] (y) {$q_3\;\tikz[baseline=(n.base)]{
\node[circle, draw=red, fill=red!20, inner sep=0.5pt, font=\scriptsize] (n) {2};
}$};
\node[state,initial,initial text=, left=1cm of b,minimum size=2mm, inner sep=1pt] (x) {$q_1\;\tikz[baseline=(n.base)]{
\node[circle, draw=blue, fill=blue!20, inner sep=0.5pt, font=\scriptsize] (n) {1};
}$};

\node[rect, right of=y,minimum size=2mm, inner sep=5pt] (e) {$(q_3, \{q_2,q_3\})\;\tikz[baseline=(n.base)]{
\node[circle, draw=green, fill=green!20, inner sep=0.5pt, font=\scriptsize] (n) {0};
}$};

%q_1 choices
\draw[->] (x) edge[bend left=30] node[above=0.2cm] {%$\lnot r$
} (b);
\draw[->] (c) edge[bend right] node[above=0.2cm] {%$\lnot g$
} (x);
\draw[->] (c) edge[bend left=15] node[above=0.2cm] {%$\lnot g$
} (w);
\draw[->] (b) edge[] node[above] {%$g$
} (z);
\draw[->] (x) edge[] node[below=0.2cm] {%$r$
} (a);
\draw[->] (a) edge[] node[above=0.2cm] {%$g$
} (z);
\draw[->] (a) edge[] node[above] {%$\lnot g$
} (y);

%q_2 choices
\draw[->] (z) edge[] node[above=0.2cm] {} (c);
\draw[->] (z) edge[] node[right] {} (d);
\draw[->] (b) edge[] node[above] {} (x);
\draw[->] (d) edge[] node[right] {} (y);
\draw[->] (d) edge[] node[right] {} (w);
\draw[->] (w) edge[bend right] node[right] {} (f);
\draw[->] (f) edge[bend right] node[right] {} (w);

%q_3 choices
\draw[->] (y) edge[] node[above=0.2cm] {} (e);
\draw[->] (e) edge[bend right] node[below=0.2cm] {} (y);
\draw[->] (e) edge[] node[above=0.2cm] {} (z);
\end{tikzpicture}
}
     \end{center}
   
The next result states that the reduction from $\mathcal{G}$ to $P_\mathcal{G}$ is correct. 

     \begin{proposition}\label{claim:GtoP_G}
         Given a two-player Hanoi omega-game \(\mathcal{G}\) with an arbitrary winning condition \(Acc \subseteq [d]^\omega\), the following holds:
         \begin{enumerate}
               \item Player \textsf{In} has an $Acc$-winning strategy in \(\mathcal{G}\) iff he has an $Acc'$-winning strategy in \(P_{\cal G}\).
    \item Player~\textsf{Out} has an $Acc$-winning strategy in \(\mathcal{G}\) iff she has an $Acc'$-winning strategy in \(P_{\cal G}\).
         \end{enumerate}
     \end{proposition}

     \begin{proof} 
     We prove Item~1. Item~2 follows symmetrically.

        A history in $\mathcal{G}$ is a finite sequence $h = v_1^\inp v_1^\outp \dots v_k^\inp v_k^\outp$ for some $k \in \mathbb{N}$. The history $h$ is \emph{compatible} with a strategy $\lambda_\inp$ of Player~\textsf{In} in $\cal G$ if for all $i \geq 1$,  
\[
v_i^\inp = \lambda_\inp\big((v_1^\inp \cup v_1^\outp)\dots (v_{i-1}^\inp \cup v_{i-1}^\outp)\big) \quad \text{and} \quad v_i^\outp \in \val(\ap_{\outp}).
\]

A history in $P_{\mathcal{G}}$ is a finite sequence $h' = q_0 (q_0,S_0) q_1 (q_1,S_1) \dots q_k$ for some $k \geq 0$
with $S_i$ is $q_i$-good and $q_{i+1} \in S_i$ for all $i \in \{0,\ldots,k-1\}$. The history $h'$ is compatible with a strategy $\lambda'_\inp$ of Player~\textsf{In} in $P_{\cal G}$ if for all $i \geq 1$,
\[
(q_i,S_i) = \lambda'_\inp\big(q_0 (q_0,S_0) \dots q_{i-1}(q_{i-1},S_{i-1})q_i\big) \quad \text{and} \quad q_{i+1} \in S_i.
\]
\medskip
\noindent\textbf{From $\mathcal{G}$ to $P_{\mathcal{G}}$.} Let $\sigma_\inp$ be a strategy for Player~\textsf{In} in $\mathcal{G}$. We construct a strategy $\sigma'_\inp$ in $P_{\mathcal{G}}$ inductively on histories, maintaining the following invariant: for any history
\[
h' = q_0 (q_0,S_0) q_1 (q_1,S_1) \dots q_{k-1} (q_{k-1},S_{k-1}) q_k
\]
in $P_{\mathcal{G}}$ that is compatible with $\sigma'_\inp$, there exists a history
\[
h = v_1^\inp v_1^\outp \dots v_k^\inp v_k^\outp
\]
in ${\mathcal{G}}$ that is compatible with $\sigma_\inp$ such that for all $i \in \{0,\dots,k-1\}$, $q_{i+1} = \delta(q_i, v_{i+1}^\inp \cup v_{i+1}^\outp)$, and the color sequence of $h$ is obtained from that of $h'$ by removing $0$ between consecutive colors.

The base case is when the history in $P_{\mathcal{G}}$ is $h' = q_0$, which is the initial state, and the corresponding history in $\cal G$ compatible with $\sigma_\inp$ is $h=\varepsilon$. We now define the strategy $\sigma'_\inp$ on $h'$. Let $v^\inp = \sigma_\inp(\varepsilon)$ and $S = \{ \delta(q_0, v^\inp \cup v^\outp) \mid v^\outp \in \val(\ap_\outp)\}$. We set $\sigma'_\inp(q_0) = (q_0,S)$.
Since $S$ is $q_0$-good, this is a valid move.

For the inductive step, let $h' = q_0 (q_0,S_0)\dots q_{k-1} (q_{k-1},S_{k-1}) q_k$ be a history in $P_{\mathcal{G}}$ compatible with $\sigma'_\inp$ satisfying the invariant, and let $h = v_1^\inp v_1^\outp \dots v_k^\inp v_k^\outp$ be the corresponding history in $\mathcal{G}$. Define
\[
v^\inp = \sigma_\inp(h),
\qquad
S = \{ \delta(q_k, v^\inp \cup v^\outp) \mid v^\outp \in \val(\ap_\outp)\}.
\]
We set $\sigma'_\inp(h') = (q_k,S)$. Since $S$ is $q_k$-good, this is a valid move in $P_{\mathcal{G}}$. We show that the invariant is preserved. Let Player~\textsf{Out} choose some successor $q_{k+1} \in S$. By definition of $S$, there exists a valuation $v^\outp$ such that $q_{k+1} = \delta(q_k, v^\inp \cup v^\outp)$. Appending $h$ with $v^\inp v^\outp$ yields a longer history in $\mathcal{G}$ compatible with $\sigma_\inp$, and the corresponding sequence in $P_{\mathcal{G}}$ is extended by $(q_k,S) q_{k+1}$. Hence the invariant is preserved.

Let $\pi'$ be any play in $P_{\mathcal{G}}$ compatible with $\sigma'_\inp$. By the invariant, every finite prefix of $\pi'$ corresponds to a finite history in $\mathcal{G}$ compatible with $\sigma_\inp$. Hence, $\pi'$ induces a play $\pi$ in $\mathcal{G}$ compatible with $\sigma_\inp$ and the color sequence of $\pi'$ is obtained from that of $\pi$ by inserting $0$ between consecutive colors. Since $\sigma_\inp$ is winning, $\pi$ satisfies $Acc$. Hence $\pi'$ satisfies $Acc'$. Therefore, $\sigma'_\inp$ is winning.

\medskip
\noindent\textbf{From $P_{\mathcal{G}}$ to $\mathcal{G}$.} Let $\sigma'_\inp$ be a strategy for Player~\textsf{In} in $P_{\mathcal{G}}$. We construct a strategy $\sigma_\inp$ in $\mathcal{G}$ inductively on histories, maintaining the following invariant: for any history
\[
h = v_1^\inp v_1^\outp \dots v_k^\inp v_k^\outp
\]
in $\mathcal{G}$ that is compatible with $\sigma_\inp$, there exists a history 
\[
h' = q_0 (q_0,S_0) q_1 (q_1,S_1) \dots q_{k-1} (q_{k-1},S_{k-1}) q_k
\]
in $P_{\mathcal{G}}$ compatible with $\sigma'_\inp$ such that:
\begin{itemize}
    \item for all $i \in \{0,\dots,k-1\}$, $q_{i+1} = \delta(q_i, v_{i+1}^\inp \cup v_{i+1}^\outp)$,
    \item for all $i \in \{0,\dots,k-1\}$, the set $S_i$ is $q_i$-good and $q_{i+1} \in S_i$.
    \item the color sequence of $h'$ is obtained from that of $h$ by inserting $0$ between consecutive colors.
\end{itemize}
Let $h = v_1^\inp v_1^\outp \dots v_k^\inp v_k^\outp$ be a history in ${\mathcal{G}}$ satisfying the invariant, and let  the corresponding history in $P_{\mathcal{G}}$ compatible with $\sigma'_\inp$ be $h'=q_0 (q_0,S_0) q_1 (q_1,S_1) \dots q_{k-1} (q_{k-1},S_{k-1}) q_k$. Define $(q,S) = \sigma'_\inp(h')$. Note that $q=q_k$ since $h'$ ends in $q_k$. Since $S$ is $q_k$-good, there exists $v^\inp$ such that $S = \{ \delta(q_k, v^\inp \cup v^\outp) \mid v^\outp \in \val(\ap_\outp)\}$. We set $\sigma_\inp(h) = v^\inp$. 
 Let Player~\textsf{Out} choose any valuation $v^\outp$, and let $q_{k+1} = \delta(q_k, v^\inp \cup v^\outp)$. By the choice of $v^\inp$, we have $q_{k+1} \in S$. Hence the sequence $q_0 (q_0,S_0) \dots q_k (q_k,S) q_{k+1}$
is a valid extension of $h'$ in $P_{\mathcal{G}}$, and the invariant continues to hold.

Let $\pi$ be any play in $\mathcal{G}$ compatible with $\sigma_\inp$. By the invariant, every finite prefix of $\pi$ corresponds to a finite history in $P_{\mathcal{G}}$ compatible with $\sigma'_\inp$. Hence, $\pi$ induces a play $\pi'$ in $P_{\mathcal{G}}$ that is compatible with $\sigma'_\inp$ and the color sequence of $\pi'$ is obtained from that of $\pi$ by inserting $0$ between consecutive colors. Since $\sigma'_\inp$ is winning, $\pi'$ satisfies $Acc'$. By the definition of $Acc'$, this implies that the sequence of colors induced by $\pi$ satisfies $Acc$. Therefore, $\pi$ is winning in $\mathcal{G}$, and $\sigma_\inp$ is a winning strategy.

     \end{proof}

Observe that the arena of \(P_{\cal G}\) is exponentially larger than that of \(\mathcal{G}\). Therefore to obtain upper bounds in the polynomial hierarchy, the game \(P_{\cal G}\) cannot be constructed explicitly. Instead, the main idea is to non-deterministically guess a strategy "of small size" for Player \textsf{In} or Player \textsf{Out} in \(P_{\cal G}\), and to check if it is winning. We start with Streett winning conditions.

\begin{lemma}\label{HOGStreetpi_2}
    Hanoi omega-games with Streett winning condition belongs\footnote{We recall that the considered decision problem is that of deciding if Player \textsf{Out} has a winning strategy.} to \(\boldsymbol{\Pi_2}\).
\end{lemma}

\begin{proof}
    Let ${\cal G} = ({\cal A},\ap_\inp,\ap_\outp, Acc)$ be an HOG with Streett winning condition $Acc$ and let  ${\cal A} = (Q, q_0, \ap_\inp \cup \ap_\outp, \delta, \cal C)$ be the underlying HOA.
    By \Cref{prop:HOG_determinacy}, to show that checking whether Player \textsf{Out} has a winning strategy in $\cal G$ is in $\boldsymbol{\Pi_2}$, it suffices to show that checking whether Player \textsf{In} has a winning strategy in $\cal G$ is in \(\boldsymbol{\mathbf{co}\Pi_2} = \boldsymbol{\Sigma_2}\). By \Cref{claim:GtoP_G}, Player~\textsf{In} has a winning strategy in $\cal G$ iff Player \textsf{In} has a winning strategy in $P_{\cal G}$. Note that the arena of $P_{\cal G}$ is exponentially larger than that of $\cal G$, however we avoid constructing $P_{\cal G}$ and show that checking whether Player \textsf{In} has a winning strategy in $P_{\cal G}$ is in $\boldsymbol{\Sigma_2}$.
    
    In classical games, Streett winning condition admit memoryless
winning strategies for Player~{\sf In}.  
Thus, to decide whether Player~{\sf In} wins in \(P_{\cal G}\), it suffices to
guess a memoryless strategy for Player~{\sf In}.  
Such a strategy consists, for each state \(q\), of a valuation
\(v_q^\inp \in \val(\ap_\inp)\) and a subset \((q,S_q)\). Hence, the strategy can therefore be represented in polynomial space. Now, to verify that this guess is correct, two conditions must be checked:
\begin{enumerate}
\item \emph{The guessed strategy is winning for Player~{\sf In}.}
From the memoryless strategy, construct the induced graph, which is a
subgraph of \(P_{\cal G}\) of size \(2|{\cal G}|\).  
One must check that in this graph there is no infinite path that is winning for Player~{\sf Out} --- It suffices to check that no reachable strongly connected component satisfies the Streett condition for Player~\textsf{Out}. This can be done in polynomial time by computing SCCs and checking the Streett condition on each of them. 

\item \emph{Each set \(S_q\) in $(q,S_q)$ is \(q\)-good.}
For each \(q \in Q\), let $\textit{Formulas}(q) \;=\;
\{\, \varphi \mid \delta(q,\varphi)\ \text{is defined} \,\}$.
It suffices to call an oracle checking the satisfiability of the following
\(\Pi_1\)-formula:
\begin{equation}\label{eq:q-good}
%\bigwedge_{q\in Q}\ 
\forall \boldsymbol{x}_\outp\ 
\bigwedge_{\substack{\varphi \in \textit{Formulas}(q)\\
        \delta(q,\varphi)\notin S_q}}
\neg \varphi\bigl(v_q^\inp,\boldsymbol{x}_\outp\bigr).
\end{equation}
\end{enumerate}

Checking that Player~{\sf In} wins in $P_{\cal G}$ (hence $\cal G$) therefore lies in
\(\NP^{\boldsymbol{\Pi_1}} = \boldsymbol{\Sigma_2}\)
 since \(\NP^{\boldsymbol{\Pi_1}} = \NP^{\boldsymbol{\mathbf{co}\Pi_1}}
= \NP^{\boldsymbol{\Sigma_1}}\).
Since HOG are determined for $\omega$-regular winning conditions (See \Cref{prop:HOG_determinacy}), checking whether Player~{\sf Out} wins in $\cal G$ belongs to
\(\boldsymbol{\mathbf{co}\Sigma_2} = \boldsymbol{\Pi_2}\).
\end{proof}

As a consequence of \Cref{lemma:HOGpi_2hard} and \Cref{HOGStreetpi_2}, we get the following result.
\begin{theorem}
    The Hanoi omega-games with Streett/ parity/ B\"uchi / co-B\"uchi/ reachability/ safety winning conditions are $\boldsymbol{\Pi_2}$-complete.
\end{theorem}

In the case of HOGs with Rabin winning conditions, although the problem is $\boldsymbol{\Pi_2}$-hard (See \Cref{lemma:HOGpi_2hard}), the proof of \Cref{HOGStreetpi_2} for Streett conditions does not extend since in classical Rabin games, memoryless winning strategies need not exist for Player~\textsf{In}. Instead, memoryless strategies exist for Player~\textsf{Out}. However, the subgraph induced by a memoryless strategy in $P_{\mathcal{G}}$ is still of exponential size, because positions of Player \textsf{Out} are pairs in $Q\times 2^Q$. However, we show that if Player \textsf{Out} wins, she can win with a "less-than-memoryless" strategy which can be represented using  polynomial space only. This allows us to 
show a $\boldsymbol{\Sigma_3}$ upper bound for HOGs with Rabin winning conditions.

\begin{lemma}\label{HOGRabinsigma_3}
    Hanoi omega-games with Rabin winning condition belongs to  \(\boldsymbol{\Sigma_3}\).
\end{lemma}

\begin{proof}
     Let ${\cal G} = ({\cal A},\ap_\inp,\ap_\outp, Acc)$ be an HOG with Rabin winning condition $Acc$ and let  ${\cal A} = (Q, q_0, \ap_\inp \cup \ap_\outp, \delta, \cal C)$ be the underlying deterministic state-based HOA. By \Cref{claim:GtoP_G}, Player~\textsf{In} has a winning strategy in $\cal G$ iff Player \textsf{In} has a winning strategy in the classical game $P_{\cal G}$. In classical games, Rabin objectives admit memoryless winning strategies for Player~\textsf{Out}. However, in our setting, it is not sufficient to guess a memoryless strategy for Player~\textsf{Out} in $P_{\cal G}$, since the subgraph of $P_{\cal G}$ induced by such a strategy may still have exponentially large size as Player~{\sf Out} owns exponentially many states in $P_{\cal G}$.

Our approach is to exploit the existence of a memoryless winning strategy to derive a more structured form of strategy, which we call a \emph{super-memoryless} strategy, as formalized in the following claim.

\begin{claim}\label{claim:super-memoryless}
    If Player \textsf{Out} has a winning strategy $\lambda_\outp$ in $P_{\cal G}$ for Rabin, then it has a \emph{super-memoryless} winning strategy $\lambda'_\outp$, in the sense that for any state $q$ of the arena, there exists a linear order $\le_q$ on $Q$ such that in $P_{\cal G}$, whenever Player~\textsf{In}, from state $q \in V_{\inp}$, moves to some subset $S \subseteq Q$, then $\lambda'_\outp$, from $(q,S)$, selects the $\le_q$-minimal element of $S$.
\end{claim}

Note that we call this strategy super-memoryless  as it is a restricted memoryless strategy in which, for each state $(q,S)$, the choice is determined by a fixed
linear order $\le_q$ on $Q$, and the strategy always selects the minimal element in $S$ with respect to $\le_q$.
In particular, for any $q$-good set $S$ and all $q$-good subset $S'\subseteq S$ that contains the $\leq_q$-minimal element of $S$, a super-memoryless strategy picks the same state from $(q,S)$ and $(q,S')$, even though there could be exponentially many such subsets $S'$.

\begin{claimproof}
Let $P_{\cal G} = (V_\inp \cup V_\outp, V_\inp=Q, V_\outp=Q \times 2^Q,E,C, \mathcal{C}', Acc')$.
    Assume that Player~\textsf{Out} has a winning strategy for the Rabin objective in the game $P_{\cal G}$.
Since $P_{\cal G}$ is a classical Rabin game,  Player~\textsf{Out} admits a \emph{memoryless} winning strategy. Let $\lambda_{\outp}$ be such a strategy.

Fix a state $q \in V_{\inp}$ and let $\mathcal{X}_q$ be the set of all $q$-good sets. For each $S \in \mathcal{X}_q$, let $p_S = \lambda_{\outp}(q,S) \in V_{\inp}$ denote the state in $Q$ selected by the strategy $\lambda_{\outp}$ at the position $(q,S)$.

We construct a linear order $\le_q$ on $Q$ iteratively as follows: Initialize $\mathcal{X} := \mathcal{X}_q$ and let the order be initially empty. While $\mathcal{X} \neq \emptyset$, pick any set $S \in \mathcal{X}$, and declare $p_S$ to be the next element in the order $\le_q$, i.e., the smallest among the elements not yet ordered. Then remove from $\mathcal{X}$ all subsets $S' \in \mathcal{X}$ such that $p_S \in S'$. Iterate this procedure until $\mathcal{X}$ becomes empty. After this process terminates, extend $\le_q$ arbitrarily to a total order on $Q$ by placing the remaining elements at the end in an arbitrary way. Let $\lambda'_{\outp}$ be the strategy that selects, at each position $(q,S) \in V_\outp$, the $\le_q$-minimal element of $S$, i.e., $\lambda'_{\outp}(q,S) = \min_{\le_q}(S)$. 

We show that $\lambda'_{\outp}$ is winning. 
Let $\pi' = q_0 (q_0,S_0) q_1 (q_1,S_1) \dots$ be any play in $P_{\cal G}$ compatible with $\lambda'_{\outp}$. We show that there exists a play of the form $\pi = q_0 (q_0,\widetilde{S}_0) q_1 (q_1,\widetilde{S}_1) \dots$ compatible with $\lambda_{\outp}$.

The construction proceeds inductively along the positions of the play. Both plays start at the same initial state. Suppose that at position $(q_i,S_i)$ in $\pi'$, the corresponding position in $\pi$ is $(q_i,\widetilde{S}_i)$, and both plays have reached the same state $q_i$. If $\lambda'_{\outp}(q_i,S_i) = \lambda_{\outp}(q_i,S_i)$, then we set $\widetilde{S}_i = S_i$, and both plays choose the same successor $q_{i+1}$. Otherwise, let $p = \lambda'_{\outp}(q_i,S_i) \neq \lambda_{\outp}(q_i,S_i)$. By construction of the order $\le_{q_i}$, there exists a set $S' \in \mathcal{X}_{q_i}$ such that $p = \lambda_{\outp}(q_i,S')$ and $p \in S'$. We then set $\widetilde{S}_i = S'$. In this case, $\pi$ moves from $(q_i,S')$ to $p$ using $\lambda_{\outp}$, while $\pi'$ moves from $(q_i,S_i)$ to $p$ using $\lambda'_{\outp}$. Thus both plays again reach the same successor state $q_{i+1} = p$.
Proceeding inductively, we obtain a play $\pi$ compatible with $\lambda_{\outp}$ such that the sequence of states visited in $Q$ is identical in $\pi$ and $\pi'$. Hence both plays induce the same sequence of colors.

Thus, every play $\pi'$ compatible with $\lambda'_{\outp}$ can be simulated by a play $\pi$ compatible with $\lambda_{\outp}$ that visits exactly the same sequence of states in $Q$, and hence the projected color sequence remains the same. Since $\lambda_{\outp}$ is winning, $\pi$ satisfies the Rabin objective, and therefore so does $\pi'$. Therefore, $\lambda'_{\outp}$ is a winning strategy. \qedhere
\end{claimproof}
Using \Cref{claim:super-memoryless}, to decide whether
Player~\textsf{Out} has a winning strategy in $P_{\cal G}$, it
suffices to guess a super-memoryless strategy, i.e., to guess, for
each state $q$ of the HOG, a linear order $\le_q$. This strategy can
therefore be represented in polynomial space. To verify that the
guessed strategy is winning, we define a directed graph $H$ over $Q$
as follows: there is an edge $(q,q')$ if there exists a $q$-good set
$S \subseteq Q$ such that $q'$ is the $\le_q$-minimal element of
$S$. Constructing the graph $H$ requires querying a
$\boldsymbol{\Sigma_2}$ oracle --- for a pair $(q,q')$, check the
existence of a subset $S$ such that $q'$ is the minimal element of $S$
with respect to $\le_q$ and $S$ is $q$-good (i.e., does there exist a
valuation $v_q^{\inp} \in \ap_{\inp}$ that satisfies the formula given
in \Cref{eq:q-good}) --- which  overall  reduces to checking the
satisfiability of the following $\boldsymbol{\Sigma_2}$ formula:
$$
\exists x_{q_1}\dots\exists x_{q_n}\exists \boldsymbol{x_\inp}\forall
\boldsymbol{x_\outp} \bigwedge_{(q',\varphi)\text{
    s.t. }\delta(q,\varphi)=q'}
\varphi\bigl(\boldsymbol{x_\inp},\boldsymbol{x_\outp}\bigr)\rightarrow
x_{q'}.
    $$
    where $Q = \{q_1,\dots,q_n\}$ and the $x_{q_i}$ are fresh
    variables (they encode a set $S$). Finally, we check that all infinite paths in $H$ satisfy the Rabin condition. This can be done in polynomial time by verifying that no reachable strongly connected component violates the Rabin condition. Thus, the verification can be performed in $\NP^{\Sigma_2}$, and therefore the problem lies in $\boldsymbol{\Sigma_3}$.
\end{proof}

Next, it remains to establish the complexity of HOG with Muller and Emerson-Lei acceptance condition. The classical Muller and EL games are known to be \PSPACE-complete~\cite{PaulHunterEL}. We show that HOG Muller and EL games are also \PSPACE-complete. The \PSPACE \ hardness follows by a reduction from the corresponding classical games, while membership in \PSPACE \ is obtained by adapting McNaughton’s algorithm for solving Muller games \cite{MCNAUGHTON1993} to the HOG setting.

\subparagraph*{Classical Muller and Emerson-lei Games.}Consider a Muller game on a finite arena $\mathcal{G} = (V, V_\inp, V_\outp, E, C, {\cal C}, {\cal F})$, where each vertex $v \in V$ has a color ${\cal C}(v) \in C$, and Player~{\sf Out} wins if the set of colors appearing infinitely often along a play belongs to ${\cal F} \subseteq 2^C$.

McNaughton's algorithm \cite[Theorem 4.1]{MCNAUGHTON1993} gives a polynomial space algorithm to compute the winning regions of $\cal G$. The algorithm recursively computes winning regions $W_\outp$ and $W_\inp$ of Player~{\sf Out} and Player~{\sf In}, respectively, as follows. For each vertex $v \in V$, define \emph{avoidance sets} $X(v)$ and $Y(v)$ as the sets of vertices from which Player~{\sf Out} or Player~{\sf In}, respectively, can avoid visiting $v$. The subgames $\mathcal{G}[X(v)]$ (whose arena is defined by the subgraph induced by the subset of vertices $X(v)$) and $\mathcal{G}[Y(v)]$ are recursively solved to obtain the respective winning regions $W_\outp(v)$ and $W_\inp(v)$. Let $E$ (resp. $A$) be the sets of vertices from which Player~{\sf Out} (resp. Player~{\sf In}) can force the play to reach some $W_\outp(v)$ (resp. $W_\inp(v)$) for some vertex $v$. If $E \cup A \neq \emptyset$, then the winning regions $W_\outp'$ and $W_\inp'$ on the smaller arena induced by $V' = V \setminus (E \cup A)$ are computed, and the winning regions $W_\outp = E \cup W_\outp'$ and $W_\inp = A \cup W_\inp'$ are returned for the original game. 

If $E \cup A = \emptyset$, neither player can force the play to any previously computed winning subarena. In this case, the winner is determined globally: let $S$ be the set of colors occurring in the remaining arena. If $S \in {\cal F}$ , then Player~{\sf Out} wins from every vertex in the arena, otherwise, Player~{\sf In} wins. The correctness of this step is proved in Lemma~2 of \cite{MCNAUGHTON1993}. Strategies are obtained by forcing the play into subarenas where a player can win, or, in the base case, cycling through all colors to satisfy the Muller condition.

The recursion depth is at most $|V|$, since each recursive call is on a strictly smaller subarena. At each level, avoidance sets and intermediate winning regions are stored, requiring $O(|V|)$ space. So the space requirements are bounded above by $O(|V|^2)$, and hence computing winner regions in Muller games is in \PSPACE.

The above algorithm is extended to Emerson-Lei games in \cite{PaulHunterEL, hunterthesis}, where the winning condition $\mathcal{F}$ is given by an Emerson-Lei Boolean formula over colors. The only modification is in the global check: instead of testing whether $S \in \mathcal{F}$, one checks whether $S$ satisfies the Emerson-Lei formula. This check is linear in the size of the formula and is performed at most once per recursion level, yielding a polynomial space procedure overall.

\subparagraph*{Adapting to Muller and Emerson-Lei HOG.} We adapt McNaughton’s algorithm and its Emerson-Lei extension to solve HOG with Muller and Emerson-Lei acceptance conditions. Consider an HOG with state space $Q$. For each state $q \in Q$, the avoidance sets $X(q)$ and $Y(q)$ are computed by solving safety HOG, which lies in \PSPACE \ (precisely in $\Pi_2$),  and are invoked at most ${\cal O}(|Q|)$ times. Similarly, the computation of the sets $E$ and $A$ from the local winning regions $W_\outp(q)$ and $W_\inp(q)$ is performed using reachability HOG, also in \PSPACE, again at most ${\cal O}(|Q|)$ times. Consequently, the algorithm computes the winning regions $W_\outp$ and $W_\inp$ for Player~{\sf{Out}} and Player~{\sf{In}}, respectively. If the initial state belongs to $W_\outp$, then Player~{\sf{Out}} has a winning strategy. Since the recursion depth is linear and all subprocedures run in polynomial space, the entire procedure remains in \PSPACE. Therefore, we get the following result.
\begin{theorem}
    HOGs with Emerson-Lei and Muller acceptance are \PSPACE-complete.
\end{theorem}

\section{Symbolic Games}

HOA can be seen as a subclass of the more general class of \emph{symbolic automata}~\cite{symbauto}, which extend finite automata to infinite data domains $\mathcal{D}$. In particular, the transitions of a symbolic automaton are guarded by formulas from a decidable first-order theory over $\mathcal{D}$, e.g., the theory of natural numbers with addition (aka Presburger arithmetic). The free variables in the guards on transitions form a (common) tuple of $n$ variables, so that the language recognized by such an automaton is a subset of words over alphabet $\mathcal{D}^{n}$. It is easily seen that HOA automata are a particular case of symbolic automata over the Boolean data domain. Likewise, HOG can be generalized to \emph{symbolic games}, whose underlying arena is given by a deterministic symbolic automaton. 

As a matter of fact, the algorithms for solving HOGs, presented in Section~\ref{sec:hog}, can be used to solve symbolic games, modulo calling some oracle for checking satisfiability of formulas of the form $\exists^*\forall^*\varphi$, for $\varphi$ formulas occurring on transitions. In particular, the reduction to solving a two-player (classical) game $P_{\cal G}$, proved in Proposition~\ref{claim:GtoP_G}, carries over to symbolic games, so that symbolic games are decidable (for $\omega$-regular winning conditions) as soon as this reduction is computable (see Theorem~\ref{th:decidability} below). For Streett conditions (and therefore parity, safety, reachability, (co)Büchi conditions), a more precise complexity analysis yields the upper-bound $\coNP^{\boldsymbol{C}}$ for solving symbolic games where the satisfiability problem of formulas of the form $\exists^*\forall^*\varphi$ is in some complexity class $\boldsymbol{C}$ (see Lemma~\ref{lem:StreettSymbolic}). For Rabin objectives, the ideas for HOG developed in Lemma~\ref{HOGRabinsigma_3} still apply, allowing us to derive the upper bound $\NP^{\boldsymbol{C}}$ (see Lemma~\ref{lem:RabinSymbolic}). We now define symbolic automata and games and give the main arguments to prove the aforementioned results.

\subparagraph*{First-order formulas} A signature $\sigma = (R,F,C)$ consists of countable sets of relation, function, and constant symbols respectively. Let $\cal X$ be a countable set of variables. The set of \emph{terms} over $\sigma$ is defined inductively: every $x \in {\cal X}$ and $c \in C$ is a term, and if $f \in F$ is $n$-ary and $t_1,\dots,t_n$ are terms, then $f(t_1,\dots,t_n)$ is a term. The set of \emph{first-order formulas} over $\sigma$, denoted by $FO[\sigma]$, is defined by the grammar:
\[
\varphi ::= r(t_1,\dots,t_n) \mid t_1 = t_2 \mid \neg \varphi \mid (\varphi \wedge \varphi) \mid (\varphi \vee \varphi) \mid \exists x\, \varphi \mid \forall x\, \varphi,
\]
where $r \in R$ is an $n$-ary relation symbol and $t_1,\dots,t_n$ are terms over $\sigma$, and $x \in {\cal X}$ is a variable.
A \emph{$\sigma$-structure} $\cal M$ consists of a non-empty domain $\cal D$ and an interpretation mapping each $r \in R$ to a relation on $\cal D$, each $f \in F$ to a function on $\cal D$, and each $c \in C$ to an element of $\cal D$. For an $FO[\sigma]$ formula $\varphi(x_1, \ldots, x_n)$
where $x_1, \ldots, x_n$ are free variables, we say $({\cal M},\boldsymbol{a})$ 
satisfies $\phi(x_1,\dots,x_n)$, written ${\cal M} \models \varphi(\boldsymbol{a})$, where $\boldsymbol{a}= (a_1, \ldots, a_n) \in {\cal D}^{n}$, 
if $\varphi$ evaluates to true in $\cal M$ under an assignment that maps each variable $x_i$ to $a_i$ for all $i \in [n]$ (detailed semantics of first-order logic can be found in \cite{Barwise1977,Libkin2004}).

A \emph{fragment} of $FO[\sigma]$ is a subset $\cal F$ of $FO[\sigma]$-formulas (e.g.the ~quantifier-free $FO[\sigma]$ formulas). Its $\exists^*\forall^*$-closure, denoted $\exists^*\forall^*\cal F$, is the fragment whose formulas are of the form $\exists x_1\dots\exists x_n\forall y_1\dots \forall y_m\varphi$ for $\varphi\in\cal F$ and an arbitrary subset of variables $\{x_1,\dots,x_n,y_1,\dots,y_m\}$.

\subparagraph*{${\cal F}$-automaton.} Let $\sigma$ be a signature and let ${\cal F} \subseteq FO[\sigma]$ be a fragment of first-order formula over $\sigma$. A ${\cal F}$-automaton $\cal A$ is a tuple ${\cal A} = (Q, I, \boldsymbol{x}=(x_1,\dots,x_n), \Delta, C)$ where $\boldsymbol{x} \subset {\cal X}$ is an ordered finite set of variables, $Q$ is the set of states, $I$ is the set of initial states, and $C: \Delta \rightarrow [d]$ is the coloring function with index $d$ (same as in the case of HOA), and $\Delta \subseteq Q \times {\cal F} \times Q$ consists of transitions that are labeled with formulas $\varphi \in \cal F$ whose free variables are a subset of $\boldsymbol{x}$. 
An ${\cal F}$-automaton is deterministic if $|I| = 1$, and for every state $q$ and tuple $\boldsymbol{a}$, there is at most one transition whose formula is satisfied by that tuple.
For a $\sigma$-structure $\mathcal{M}$ (with domain ${\cal D}$) and an acceptance condition $Acc \subseteq [d]^{\omega}$, 
 the language of $\mathcal{A}$ over $\mathcal{M}$, denoted $\mathcal{L}_{Acc}^{\cal M}({\cal A})$, is the set of infinite words $\boldsymbol{a_1} \boldsymbol{a_2} \cdots\in ({\cal D}^n)^\omega$ such that there exists a run
$\rho = q_1 \xrightarrow{\varphi_1} q_2 \xrightarrow{\varphi_2} \cdots$
with $q_1 \in I$, $\mathcal{M} \models \varphi_i(\boldsymbol{a}_i)$ for all $i \geq 1$, and the induced color sequence $\mathcal{C}(\rho)$ belongs to $Acc$.

\subparagraph*{${\cal F}$-games.} For any two tuples $\boldsymbol{c}=(c_1,\dots,c_k)$ and $\boldsymbol{d} = (d_1,\dots,d_l)$, we may write $(\boldsymbol{c},\boldsymbol{d})$ for the tuple $(c_1,\dots,c_k,d_1,\dots,d_l)$.
A \emph{two-player ${\cal F}$-game} is a tuple 
${\cal G} = ({\cal A},\boldsymbol{x^\inp},\boldsymbol{x^\outp},Acc)$
such that ${\cal A}$ is a deterministic ${\cal F}$-automaton with tuple of variables $\boldsymbol{x}=(\boldsymbol{x^\inp},\boldsymbol{x^\outp})$, of size $n+m$, where $n$ and $m$ are the sizes of $\boldsymbol{x^\inp}$ and $\boldsymbol{x^\outp}$ respectively, and $Acc\subseteq [d]^\omega$ is the winning condition for $d$ the index of $\cal A$.

Let fix a $\sigma$-structure $\mathcal{M}$ (with domain ${\cal D}$). The semantics of an ${\cal F}$-game ${\cal G} = ({\cal A},\boldsymbol{x^\inp},\boldsymbol{x^\outp},Acc)$ is defined with respect to $\mathcal{M}$ as follows. The game is played in turns by two players, Player \textsf{In} and Player \textsf{Out}, who successively pick valuations of $\boldsymbol{x^{\inp}}$ and $\boldsymbol{x^{\outp}}$ respectively. Player \textsf{In} plays first, picking some tuple $\boldsymbol{a^\inp_1}\in{\cal D}^n$, then Player \textsf{Out} picks some tuple $\boldsymbol{a^\outp_1}\in{\cal D}^m$, and the game continues like that \emph{ad infinitum}, forming an infinite sequence $\pi = (\boldsymbol{a^\inp_1},\boldsymbol{a^\outp_1})(\boldsymbol{a^\inp_2},\boldsymbol{a^\outp_2})\dots$ For an acceptance condition $Acc  \subseteq [d]^\omega$, Player \textsf{Out} wins this play if $\pi$ belongs to $\mathcal{L}_{Acc}^{\cal M}({\cal A})$, otherwise Player \textsf{In} wins. Finally, the game is won by Player \textsf{Out} if he has a strategy 
$\lambda_{\outp}:(\mathcal{D}^{n+m})^*\mathcal{D}^n\rightarrow \mathcal{D}^m$ to pick the tuples $\boldsymbol{a^\outp_i}$, provided a history of the tuples chosen by both players so far.
Similar to HOG, ${\cal F}$-games are also determined for $\omega$-regular winning conditions using Borel determinancy of classical games. We say that an ${\cal F}$-game is \emph{decidable} (or belongs to a complexity class ${\cal C}$) if it is possible to decide whether Player~\textsf{Out} has a winning strategy (resp.~within ${\cal C}$).  

Every ${\cal F}$-automaton $\cal A$ can be transformed into an equivalent state-based ${\cal F}$-automaton with a coloring function ${\cal C}: Q \to [d]$. This transformation is identical to the one for HOA (see \Cref{lem:to_state_based}). 
Given a ${\cal F}$-game ${\cal G} = ({\cal A},\boldsymbol{x^\inp},\boldsymbol{x^\outp},Acc)$, we assume wlog that ${\cal A}=(Q,q_0,\boldsymbol{x}= (\boldsymbol{x^\inp},\boldsymbol{x^\outp}),\delta,{\cal C})$ is deterministic and state-based ${\cal F}$-automaton. We define the classical turn-based game $P_{\cal G}$ exactly as in the HOG case: vertices are $V_\inp = Q$ and $V_\outp = Q \times 2^Q$, with edges defined via $q$-good sets, and with winning condition $Acc'$ obtained from $Acc$ by inserting a fresh color $0$ between successive colors. \Cref{claim:GtoP_G} can be extended to ${\cal F}$-games, and hence ${\cal G}$ and $P_{\cal G}$ are equivalent with respect to winning strategies for both players.

\begin{theorem}\label{th:decidability}
    Let $\cal M$ be a $\sigma$-structure and ${\cal F} \subseteq FO[\sigma]$. If the satisfiability problem for $\exists^*\forall^*\cal F$ in $\cal M$ is decidable, then ${\cal F}$-games (over $\cal M$) are decidable for $\omega$-regular winning conditions.
\end{theorem}

\begin{proof}
   We reduce solving ${\cal G}$ to solving a classical turn-based game $P_{\cal G}$.
To construct $P_{\cal G}$, it suffices to decide, for $q \in Q$ and $S \subseteq Q$, whether $S$ is $q$-good. Let $\mathit{Form}(q,S)$ be the set of formulas $\varphi$ such that $\delta(q,\varphi)\not\in S$ (if defined). 
Since ${\cal A}$ is deterministic, $S$ is $q$-good iff the formula $\exists \boldsymbol{x^\inp}\;\forall \boldsymbol{x^\outp}\;
\bigwedge_{\varphi\in \mathit{Form}(q,S)}
\neg \varphi\bigl(\boldsymbol{x^\inp},\boldsymbol{x^\outp})
$
is satisfiable in $\cal M$. By assumption, satisfiability for $\exists^*\forall^*\cal F$ is decidable. Thus, we can decide whether $S$ is $q$-good, and thus construct all edges of $P_{\cal G}$. Hence the reduction to $P_{\cal G}$ is effective.

 Since $P_{\cal G}$ is a finite turn-based classical game with an $\omega$-regular winning condition, its winner is decidable. Since ${\cal G}$ and $P_{\cal G}$ are equivalent with respect to winning strategies for both players, we conclude that the winner in ${\cal F}$-games is also decidable.
\end{proof}

Note that by determinacy of symbolic games, the previous theorem entails that we can decide if Player \textsf{In} wins (under the same assumptions as the theorem statement).

\begin{lemma}\label{lem:StreettSymbolic}
Let $\cal M$ be a $\sigma$-structure and ${\cal F} \subseteq FO[\sigma]$. If the satisfiability problem for $\exists^*\forall^*\cal F$ in $\cal M$ is in some complexity class \(\boldsymbol{C}\), then ${\cal F}$-games are in the complexity class $\coNP^{\boldsymbol{C}}$ for Streett, parity, B\"uchi, co-B\"uchi, safety and reachability winning conditions. 
\end{lemma}
 \begin{proof}
     The proof follows the same lines as for HOG (Lemma~\ref{HOGStreetpi_2}).   
     Given an ${\cal F}$-game ${\cal G} = ({\cal A},\boldsymbol{x^\inp},\boldsymbol{x^\outp},Acc)$, we reduce it to the classical game $P_{\cal G}$ preserving winning strategies for both players. As in the HOG case, we avoid constructing $P_{\cal G}$ explicitly. For all the considered winning conditions, Player~\textsf{In} admits memoryless winning strategies. Hence, the algorithm guesses a memoryless strategy and verifies it. Such a strategy consists of $(q,S_q)$ for each state $q$. The verification of the guess consists of two checks identical to the HOG setting: 1) checking whether the guessed strategy is winning for Player~\textsf{In} is in polynomial time and done exactly as for HOGs, and 2) checking that each chosen set $S_q$ is $q$-good. The latter reduces to checking the satisfiability of formulas of the form
$\exists \boldsymbol{a^\inp} \forall \boldsymbol{a^\outp}\; \varphi(\boldsymbol{a^\inp},\boldsymbol{a^\outp})
$ with $\varphi \in {\cal F}$, which lies in \(\boldsymbol{C}\) by assumption. Thus, checking whether Player~\textsf{In} wins lies in $\NP^{\boldsymbol{C}}$ and by determinacy, checking whether Player~\textsf{Out} wins lies in $\coNP^{\boldsymbol{C}}$.
 \end{proof}
 \begin{lemma}\label{lem:RabinSymbolic}
 Let $\cal M$ be a $\sigma$-structure and ${\cal F} \subseteq FO[\sigma]$. If the satisfiability problem for $\exists^*\forall^*\cal F$ in $\cal M$ is in some complexity class \(\boldsymbol{C}\), then ${\cal F}$-games belongs to the complexity class $\NP^{\boldsymbol{C}}$ for Rabin winning conditions.
\end{lemma}
\begin{proof}
    We argue as in Lemma~\ref{HOGRabinsigma_3}. 
Like the HOG case, for Rabin winning conditions,  Player~\textsf{Out} admits super-memoryless winning strategies in $P_{\cal G}$. So, we guess for each state $q$ a linear order $\le_q$ on $Q$, which defines a strategy.
To verify the guess, we construct the induced graph over $Q$: an edge $(q,q')$ exists if $q'$ is the $\le_q$-minimal element of some $q$-good set. Checking existence of such an $S$ reduces to checking satisfiability of an $\exists^*\forall^*\cal F$ formula, a problem in $\boldsymbol{C}$ by assumption. Finally, we check in polynomial time that all reachable SCCs in the induced graph satisfy the Rabin condition. Thus, the problem lies in $\NP^{\boldsymbol{C}}$.
\end{proof}

\subparagraph*{Symbolic games over $(\mathbb{N},+)$.} Consider Presburger arithmetic, i.e., the first-order theory of natural numbers with addition \cite{HaaseSG18, borosh1976bounds}. In particular, let $\cal F_{\mathbb{N},+}$ be its quantifier-free fragment. The satisfiability problem for $\exists^*\cal F_{\mathbb{N},+}$ is $\NP$-complete, whereas the $\forall^*\exists^* \cal F_{\mathbb{N},+}$ and $\exists^*\forall^* \cal F_{\mathbb{N},+}$ fragments lie in the weak exponential hierarchy~\cite{Haase14}. Hence, by \Cref{th:decidability}, $\cal F_{\mathbb{N},+}$-games  are decidable. If the number of variables in each quantifier block is fixed, then the $\forall^*\exists^*$ and $\exists^*\forall^*$ fragments become $\coNP$-complete and $\NP$-complete, respectively~\cite{HaaseSG18,Erich1988}. Consequently, solving $\cal F_{\mathbb{N},+}$-games with a fixed number of variables, is $\coNP$-hard and lies in $\boldsymbol{\Pi_2}$ for Streett, parity, B\"uchi, co-B\"uchi, reachability, and safety winning conditions. The upper bound follows from \Cref{lem:StreettSymbolic}, while $\coNP$-hardness is obtained analogously to the proof of \Cref{lemma:HOGpi_2hard}. For Rabin winning conditions, the problem is $\coNP$-hard and lies in $\boldsymbol{\Sigma_2}$ (by \Cref{lem:RabinSymbolic}).

\bibliography{references}

\end{document}